\documentclass{gGAF2e}
\usepackage{natbib}
\usepackage{color}
\usepackage{makecell}
\pdfoutput=1
 
\usepackage[normalem]{ulem}
\usepackage[center]{caption}
\usepackage{amssymb}
\usepackage{amsfonts}
\begin{document}

\jvol{00} \jnum{00} \jyear{2023} 

\markboth{\rm C.~LIU and  A.D.~CLARK}{\rm GEOPHYSICAL \&  ASTROPHYSICAL FLUID DYNAMICS}

\title{Analysing the impact of bottom friction on shallow water waves over idealised bottom topographies}

\author{Chang Liu ${\dag}$ ${\ddag}$ $^\ast$\thanks{$^\ast$Corresponding author. Email: cliu124@alumni.jh.edu
\vspace{6pt}} and Antwan~D.~Clark ${\dag}$ \\\vspace{6pt} ${\dag}$ Department of Applied Mathematics and Statistics, Johns Hopkins University \\ Baltimore,
MD 21218 USA\\ ${\ddag}$ Department of Physics, University of California, Berkeley \\
Berkeley, CA 94720 USA  \\\vspace{6pt}\received{v4.4 released July 2022} }

\maketitle

\begin{abstract}
	
Analysing the impact of bottom friction on shallow water waves over bottom terrains is important in areas including environmental and coastal engineering as well as the oceanic and atmospheric sciences. However, current theoretical developments rely on making certain limiting assumptions about these flows and thus more development is needed to be able to further generalise this behaviour. This work uses Adomian decomposition method (ADM) to not only develop semi-analytical formulations describing this behaviour, for flat terrains, but also as reverse-engineering mechanisms to develop new closed-form solutions describing this type of phenomena. Specifically, we respectively focus on inertial geostrophic oscillations and anticyclonic vortices with finite escape times in which our results directly demonstrate the direct correlation between the constant Coriolis force, the constant bottom friction, and the overall dynamics. Additionally, we illustrate elements of dissipation-induced instability with respect to constant bottom friction in these types of flows where we also demonstrate the connection to the initial dynamics for certain cases. 

\end{abstract}

\section{Introduction}
\label{sec:introduction}

\par Shallow water equations are widely employed in environmental and coastal engineering applications such as forecasting wave run-up on coastal structures caused by tides, storm surges, hurricanes, and tsunamis \citep{cushman2011introduction}. Understanding the effects of bottom friction is also key to better understand these phenomena. For example, understanding the impact of bottom shear stress has become a key component to understanding dam-break flows as well as how tsunamis propagate within shallow regimes \citep{wang20112d, tinh2021numerical}. Other areas where understanding the effects of bottom friction are important include flows within shallow estuaries \citep{parker1984frictional}, flows around isolated barriers within velocity-deficit regions  \citep{grubivsic1995effect}, and understanding the effects of jets including their various formations \citep{vasavada2005jovian,warneford2017super}. The impact of bottom friction with respect to hydrodynamic stability has also been explored, where earlier works include that of \cite{chen1997absolute} who discovered that
absolute and convective instabilities of plane turbulent wake in shallow water layers can be suppressed by the bottom friction. Recently, \cite{jin2019frictional} discovered that friction leads to instabilities when exploring the effects of discontinuous interfaces in tangential velocities where the authors note that friction-based effects are suppressed for certain Froude numbers.

\par Theorising flow effects due to bottom friction, including developing analytical solutions, are beneficial in the development, understanding, and bench marking numerical simulations of this behaviour. For example, the Shallow Water Analytic Solutions for Hydraulic and Environmental Studies (SWASHES) software library incorporated solutions related to open-channel flows \citep{macdonald1996analysis,macdonald1997analytic} as well as moving boundary shallow water flows \citep{sampson2003moving,sampson2005moving,sampson2007moving}. The steady-state results of \cite{macdonald1996analysis} and \cite{macdonald1997analytic} were employed to bench mark numerical simulations for overland flows \citep{delestre2017fullswof}. 
The work of Samspon {\em et al.} was also used to validate Godunov-based methods \citep{wang20112d,berthon2011efficient,hou20132d}, two-dimensional well-balanced numerical schemes \citep{hou20132d}, and asymptotic preserving schemes \citep{duran2015asymptotic}. Additionally, the works of \cite{macdonald1997analytic} and \cite{sampson2005moving} were used to validate variational data assimilation packages for shallow-water models \citep{couderc2013dassfow}.

\par However, these analytical developments depend on making limiting assumptions describing these flows and thus generalised formulations are needed to robustly describe the impact of bottom friction and Coriolis force on these dynamics \citep{matskevich2019exact}. To help address these limitations, we extend the work of \cite{liu2021ADMSWE} to employ Adomian decomposition method to the shallow-water equations while considering idealised bottom topographies with constant bottom friction, where our main contributions include the following. First, we establish the connections between the theoretical assumptions presented in the works of \cite{thacker1981some} and \cite{matskevich2019exact} which consider the respective cases of with and without bottom friction. Next, we develop new classes of solutions describing inertial oscillations in geostrophic flows and anticyclonic vortices with finite escape times where these results show the direct correlation between the constant Coriolis force, constant friction, and the dynamic behaviour over flat bottom topographies. 

\par The rest of this paper is organized in the following manner. Section \ref{sec:SWE} introduces the Adomian decomposition formulation of the shallow-water model, where we directly show the connection to the fundamental assumptions presented in the aforementioned works. Section \ref{sec:new_exact} provides additional derivations of solutions of shallow water waves over flat bottom terrains with corresponding linear friction, for the cases of inertial oscillations in geostrophic flows and anticyclonic vortices with finite escape times, for both short and long-term behaviour. Section \ref{sec:numerical} presents our numerical experimentation and results. Section \ref{sec:conclusion} concludes this paper, where we highlight some avenues of future exploration. 

\section{Adomian Decomposition Formulation}
\label{sec:SWE}

The non-dimensional form of the rotating shallow-water equations with an idealized bottom topography and a linear friction term are described as
\begin{subequations}
\label{eq:SWE_non_dimen}
\begin{align}
    \frac{\upartial u}{\upartial t}=&-u\frac{\upartial u}{\upartial x}-v\frac{\upartial u}{\upartial y}-\frac{1}{F^2}\frac{\upartial h}{\upartial x}+\bar{f}v-\bar{\tau} u,\\
    \frac{\upartial v}{\upartial t}=&-u\frac{\upartial v}{\upartial x}-v\frac{\upartial v}{\upartial y}-\frac{1}{F^2}\frac{\upartial h}{\upartial y}-\bar{f}u-\bar{\tau} v,\\
    \frac{\upartial h}{\upartial t}=&-\frac{\upartial }{\upartial x}[u(D+h)]-\frac{\upartial }{\upartial y}[v(D+h)],%
\end{align}
\end{subequations}

\noindent  where $u$ and $v$ are the horizontal and vertical velocity components in the $x$ and $y$ directions and $h$ is the free surface height. The horizontal and vertical lengths (including the free surface height) are normalized by their respective length scales, $L_0$ and $H_0$. Furthermore, the velocity components and the time scale $t$ are normalized by $U_0$ and $L_0/U_0$, respectively. This results in the Froude number $F=U_0/\sqrt{gH_0}$, the normalized Coriolis force $\bar{f}=fL_0/U_0$, and the normalized friction component $\bar{\tau} = \tau L_0/U_0$. This work considers idealized bottom topographies defined as

\begin{align}
    D=D_0\left(1-\frac{x^2}{L^2}-\frac{y^2}{l^2}\right),
\label{eq:idealBottomForm}
\end{align}

\noindent where flat bottom terrains are also considered by setting $D_0=0$. The total fluid depth $D+h$, follow the formulations of \citet{thacker1981some} and \citet{shapiro1996nonlinear}, where $D+h=0$ represents a moving shoreline and
$D+h<0$ represents dry regions. It is also important to mention that our explorations in this section consider flow velocities that are linearly varying spatially while the free surface height either varies linearly or in a quadratic fashion. The initial conditions are given by
\begin{align}
\label{eq:SWE_initial}
    u(x,y,0)=u_0(x,y),\;\;
    v(x,y,0)=v_0(x,y),\;\;\text{and}\;\;
    h(x,y,0)=h_0(x,y).
\end{align}

Next, $u$, $v$, and $h$ are decomposed as follows
\begin{align}
\begin{bmatrix}
    u(x,y,t)\\
    v(x,y,t)\\
    h(x,y,t)
\end{bmatrix}=\sum_{n=0}^\infty \begin{bmatrix}
    u_n(x,y,t)\\
    v_n(x,y,t)\\
    h_n(x,y,t)
\end{bmatrix},
    \label{eq:ADM_series}
\end{align}

\noindent where the initial components $(n=0)$ are defined by \eqref{eq:SWE_initial}. Thus, the recurrence relationships ($n\in \mathbb{N}$) to equation \eqref{eq:SWE_non_dimen} are given by

\begin{subequations}
\label{eq:ADM_iter_Linv}
\begin{align}
    u_{n+1}=&-{L}_t^{-1}\left\{A_n\left(u,\frac{\upartial u}{\upartial x}\right)+A_n\left(v,\frac{\upartial u}{\upartial y}\right)+\frac{1}{F^2}\frac{\upartial h_{n}}{\upartial x}-\bar{f}v_n+\bar{\tau} u_n\right\},
    \\
    v_{n+1}=&-{L}_t^{-1}\left\{A_n\left(u,\frac{\upartial v}{\upartial x}\right)+A_n\left(v,\frac{\upartial v}{\upartial y}\right)+\frac{1}{F^2}\frac{\upartial h_{n}}{\upartial y}+\bar{f}u_n+\bar{\tau} v_n\right\},
    \\
    h_{n+1}=&-{L}_t^{-1}\left\{\frac{\upartial }{\upartial x}[A_n(u,h)]+\frac{\upartial }{\upartial y}[A_n(v,h)]+\frac{\upartial }{\upartial x}[u_n D]+\frac{\upartial }{\upartial y}[v_n D]\right\},
\end{align}
\end{subequations}

\noindent where
\begin{equation}
    L_t:=\frac{\upartial (\cdot)}{\upartial t},\;\;\;L_t^{-1}:=\int_{0}^{t}(\cdot){\rm d}t',
\end{equation}

\noindent and the Adomian polynomial representing the quadratic nonlinearity is defined by \citep{adomian2013solving}
\begin{align}
    A_n(u,h):=\sum_{j=0}^n u_{j}h_{n-j}.
    \label{eq:ADM_An}
\end{align}

\noindent From \eqref{eq:ADM_An} the quadratic nonlinear terms, such as $uh$, can be approximated as  
\begin{align}
    uh=\left(\sum_{p}^{\infty}u_p \right)\left(\sum_{q}^{\infty}h_q \right)=\sum_{n}^{\infty}A_n(u,h),
    \label{eq:ADM_An_sum}
\end{align}

\noindent which leads to the following partial sums that satisfy equation \eqref{eq:SWE_non_dimen} 

\begin{equation}
 S_N(u)=\sum_{n=0}^N u_n,\;\;S_N(v)=\sum_{n=0}^N v_n,\;\;\text{and}\;\; S_N(h)=\sum_{n=0}^N h_n.
 \label{eq:partial_sum}
\end{equation}

\par Next, the following results connect the properties of the initial conditions to the behaviours of the true solutions via their partial sums. 

\begin{lemma}
Let $\{u_n(x,y,t)\}$, $\{v_n(x,y,t)\}$, $\{h_n(x,y,t)\}$ be the sequence of decomposed functions of $u$, $v$, and $h$, where their relationship is defined by \eqref{eq:ADM_iter_Linv} (for $n \in \mathbb{N}$) given an ideal parabolic topography \eqref{eq:idealBottomForm}. If the initial conditions $u_0(x,y)$, $v_0(x,y)$, $h_0(x,y)$ are defined such that 

\begin{equation}
    \frac{\upartial^2 u_0(x,y)}{\upartial x^2} = \frac{\upartial^2 v_0(x,y)}{\upartial x^2} = \frac{\upartial^3 h_0(x,y)}{\upartial x^3} = 0,
\label{PartialDerivIC_x_Cond}
\end{equation}

\begin{equation}
    \frac{\upartial^2 u_0(x,y)}{\upartial y^2} = \frac{\upartial^2 v_0(x,y)}{\upartial y^2} = \frac{\upartial^3 h_0(x,y)}{\upartial y^3} = 0,
\label{PartialDerivIC_y_Cond}
\end{equation}

\noindent and 

\begin{equation}
    \frac{\upartial^2 u_0(x,y)}{\upartial xy}=\frac{\upartial^2 v_0(x,y)}{\upartial x \upartial y}= \frac{\upartial^3 h_0(x,y)}{\upartial x^2 \upartial y} = \frac{\upartial^3 h_0(x,y)}{\upartial x\upartial y^2} = 0.
\label{PartialDerivIC_mixedxy_Cond}
\end{equation}

\noindent Then the higher order components $u_n(x,y,t)$, $v_n(x,y,t)$, $h_n(x,y,t)$ ($n \in \mathbb{N}^{+}$) also satisfy the same property, where 

\begin{equation}
    \frac{\upartial^2 u_n(x,y,t)}{\upartial x^2} = \frac{\upartial^2 v_n(x,y,t)}{\upartial x^2} = \frac{\upartial^3 h_n(x,y,t)}{\upartial x^3} = 0,
\label{PartialDeriv_x_Cond}
\end{equation}

\begin{equation}
    \frac{\upartial^2 u_n(x,y,t)}{\upartial y^2} = \frac{\upartial^2 v_n(x,y,t)}{\upartial y^2} = \frac{\upartial^3 h_n(x,y,t)}{\upartial y^3} = 0,
\label{PartialDeriv_y_Cond}
\end{equation}

\noindent and 

\begin{equation}
    \frac{\upartial^2 u_n(x,y,t)}{\upartial x\upartial y}=\frac{\upartial^2 v_n(x,y,t)}{\upartial x\upartial y}= \frac{\upartial^3 h_n(x,y,t)}{\upartial x^2\upartial y} = \frac{\upartial^3 h_n(x,y,t)}{\upartial x\upartial y^2} = 0.
\label{PartialDeriv_mixedxy_Cond}
\end{equation}

\label{ADM:PartialDerivativeCoeffLemma}
\end{lemma}

\begin{proof}
See Appendix \ref{proof_ADM:PartialDerivativeCoeffLemma}. 
\end{proof}

\begin{theorem}
\label{thm:ADM_Ansatz_Connection}
Let $\{u_n(x,y,t)\}$, $\{v_n(x,y,t)\}$, $\{h_n(x,y,t)\}$ be the sequence of decomposed functions of $u$, $v$, and $h$, where their relationship is defined by \eqref{eq:ADM_iter_Linv} (for $n \in \mathbb{N}$) given an ideal parabolic topography \eqref{eq:idealBottomForm}. If the initial conditions $u_0(x,y)$, $v_0(x,y)$, $h_0(x,y)$ are defined as \eqref{PartialDerivIC_x_Cond}  - \eqref{PartialDerivIC_mixedxy_Cond},  then the solutions of $u$, $v$, and $h$ have the same property where 

\begin{equation}
    \frac{\upartial^2 u(x,y,t)}{\upartial x^2} = \frac{\upartial^2 u(x,y,t)}{\upartial y^2} 
    = \frac{\upartial^2 u(x,y,t)}{\upartial x \upartial y} = 0,
\label{PartialCondition_u}
\end{equation}

\begin{equation}
    \frac{\upartial^2 v(x,y,t)}{\upartial x^2} = \frac{\upartial^2 v(x,y,t)}{\upartial y^2} 
    = \frac{\upartial^2 v(x,y,t)}{\upartial x \upartial y} = 0,
\label{PartialCondition_v}
\end{equation}

\noindent and 

\begin{equation}
    \frac{\upartial^3 h(x,y,t)}{\upartial x^3}
    = \frac{\upartial^3 h(x,y,t)}{\upartial x^2 \upartial y} 
    = \frac{\upartial^3 h(x,y,t)}{\upartial x \upartial y^2} 
    = \frac{\upartial^3 h(x,y,t)}{\upartial y^3} = 0.
\label{PartialCondition_h}
\end{equation}

\noindent Consequently, these solutions can be expressed as 
\begin{equation}
    {u}(x,y,t)=\tilde{u}_0(t)+\tilde{u}_x(t)x+\tilde{u}_y(t)y,
\label{uLinearCondition}
\end{equation}
    
\begin{equation}
{v}(x,y,t)=\tilde{v}_0(t)+\tilde{v}_x(t)x+\tilde{v}_y(t)y,
\label{vLinearCondition}
\end{equation}

\noindent and 

\begin{equation}
  {h}(x,y,t)=\tilde{h}_0(t)+\tilde{h}_x(t)x+\tilde{h}_y(t)y+\frac{1}{2}\tilde{h}_{xx}(t)x^2+\frac{1}{2}\tilde{h}_{yy}(t)y^2+\tilde{h}_{xy}(t)xy,  
\label{hLinearCondition}
\end{equation}

\noindent where the coefficients $\tilde{u}_0(t)$, $\tilde{u}_x(t)$, $\tilde{u}_y(t)$, $\tilde{v}_0(t)$, $\tilde{v}_x(t)$, $\tilde{v}_y(t)$, $\tilde{h}_0(t)$, $\tilde{h}_x(t)$, $\tilde{h}_y(t)$, $\tilde{h}_{xx}(t)$, $\tilde{h}_{yy}(t)$, and $\tilde{h}_{xy}(t)$ are time-dependent.
\end{theorem}

\begin{proof}
Applying Lemma \ref{ADM:PartialDerivativeCoeffLemma} to each component in \eqref{eq:ADM_series} yields \eqref{PartialCondition_u} - \eqref{PartialCondition_h}. From \eqref{PartialCondition_u}, we observe that 

\begin{equation}
   \frac{\upartial^2 u(x,y,t)}{\upartial x^2} = 0 \, \, \, \, \mbox{yields} \, \, \, \, u(x,y,t)=C_1(y,t) x+C_2(y,t),
\label{integration_constant_C1C2_Result} \nonumber
\end{equation}

\noindent where the integration constants, $C_1(y,t)$ and $C_2(y,t)$, are independent of $x$. Similarly, we have 

\begin{equation}
   \frac{\upartial^2 u(x,y,t)}{\upartial x \upartial y} = 0 \, \, \, \, \mbox{yields} \, \, \, \, C_1(y,t) = \tilde{u}_x(t)
\label{integration_constant_2_Result} \nonumber
\end{equation}

\noindent and 

\begin{equation}
\frac{\upartial^2 u(x,y,t)}{\upartial y^2} = 0 \, \, \, \, \mbox{yields} \, \, \, \, C_2(y,t)=\tilde{u}_y(t) y+\tilde{u}_0(t).
\label{integration_constantC2_Result}\nonumber
\end{equation}

\noindent Thus, \eqref{uLinearCondition} is achieved. Similar arguments can be made to achieve \eqref{vLinearCondition} and \eqref{hLinearCondition}, respectively. 
\end{proof}

\noindent Theorem \ref{thm:ADM_Ansatz_Connection} provides a direct correlation between the initial and overall dynamic behaviour, where these results were previously presented as \emph{ansatz solutions} for cases with and without bottom friction \citep{thacker1981some,shapiro1996nonlinear,sampson2003moving,sampson2005moving,sampson2007moving,matskevich2019exact,bristeau2021some}. We also note that this is preserved in the Adomian representations, given by  \eqref{eq:ADM_iter_Linv}, for they correlate to the functional Taylor series expansions about the initial conditions \eqref{PartialDerivIC_x_Cond} - \eqref{PartialDerivIC_mixedxy_Cond}. Hence, it can be shown that the Adomian decompositions equate to the functional Taylor series expansions of previous works that consider constant bottom friction like that of \citet{sampson2003moving}, \citet{sampson2005moving}, \citet{sampson2007moving}, and \citet{matskevich2019exact}.

\section{Nonlinear growing solutions induced by bottom friction}
\label{sec:new_exact}

Next, we use the construction in \eqref{eq:ADM_iter_Linv} to derive novel families of solutions describing inertial geostrophic oscillations and anticyclonic vortices for flat bottom topographies, where $D_0=0$ in \eqref{eq:idealBottomForm} which have a profound effect on oceanic and atmospheric dynamics \citep{vallis2017atmospheric}. In this section, we theorize the impacts of constant bottom friction ($\bar{\tau} \neq 0$) and the Coriolis parameter ($\bar{f}\neq 0$) with respect to this phenomenon. 

\subsection{Inertial geostrophic oscillatory behaviour}
\label{subsec:new_decaying_geostrophic}

\par For these types of flows, our analysis considers the following initial conditions

\begin{itemize}

    \item Condition I
    \begin{equation}
        u_0(x,y)=-\frac{\tilde{h}_y}{F^2\bar{f}},\;\;v_0(x,y)=\frac{\tilde{h}_x}{F^2\bar{f}},\;\;h_0(x,y)=\tilde{h}_x x+\tilde{h}_y y,
        \label{eq:IC_cond_I}
    \end{equation}
    
    \item Condition II
    \begin{equation}
        u_0(x,y)=v_0(x,y)=0,\;\;h_0(x,y)=\tilde{h}_x x+\tilde{h}_y y,
        \label{eq:IC_cond_II}
    \end{equation}
    
    \item Condition III
    \begin{equation}
        u_0(x,y)=v_0(x,y)=0,\;\;h_0(x,y)=\tilde{h}_x x,
        \label{eq:IC_cond_III}
    \end{equation}
    
    \item Condition IV
    \begin{equation}
        u_0(x,y)=v_0(x,y)=0,\;\;h_0(x,y)=\tilde{h}_y y,
        \label{eq:IC_cond_IV}
    \end{equation}
\end{itemize}

\noindent where $\bar{f}\neq 0$ is the constant Coriolis parameter and $\tilde{h}_x$ and $\tilde{h}_y$ are the respective constant free surface gradient in the $x$ and $y$ directions. We begin with the following lemma that describes the connection between the initial behaviour, given by \eqref{eq:IC_cond_I} - \eqref{eq:IC_cond_IV}, and the decomposition of $u$, $v$ and $h$. 

\begin{lemma}
Let $\{u_n(x,y,t)\}$, $\{v_n(x,y,t)\}$, $\{h_n(x,y,t)\}$ be the sequence of decomposed functions of $u$, $v$, and $h$ such that their relationship is defined by \eqref{eq:ADM_iter_Linv}. If $D= 0$ and the initial conditions $u_0(x,y)$, $v_0(x,y)$, $h_0(x,y)$ satisfy the following properties 

\begin{equation}
    \frac{\upartial u_0(x,y)}{\upartial x} = \frac{\upartial v_0(x,y)}{\upartial x} = \frac{\upartial^2 h_0(x,y)}{\upartial x^2} = 0,
\label{PartialDerivIC_x_Cond_sin}
\end{equation}

\begin{equation}
    \frac{\upartial u_0(x,y)}{\upartial y} = \frac{\upartial v_0(x,y)}{\upartial y} = \frac{\upartial^2 h_0(x,y)}{\upartial y^2} = 0,
\label{PartialDerivIC_y_Cond_sin}
\end{equation}

\noindent and 

\begin{equation}
    \frac{\upartial^2 h_0(x,y)}{\upartial x \upartial y} = 0,
\label{PartialDerivIC_mixedxy_Cond_sin}
\end{equation}

\noindent then their higher order components $u_n(x,y,t)$, $v_n(x,y,t)$, $h_n(x,y,t)$ ($n\in \mathbb{N}^{+}$) satisfy the following properties  

\begin{equation}
    \frac{\upartial u_n(x,y,t)}{\upartial x} = \frac{\upartial v_n(x,y,t)}{\upartial x} = \frac{\upartial h_n(x,y,t)}{\upartial x} = 0,
\label{PartialDeriv_x_Cond_sin}
\end{equation}

\noindent and 

\begin{equation}
    \frac{\upartial u_n(x,y,t)}{\upartial y} = \frac{\upartial v_n(x,y,t)}{\upartial y} = \frac{\upartial h_n(x,y,t)}{\upartial y} = 0.
\label{PartialDeriv_y_Cond_sin}
\end{equation}

\label{ADM:PartialDerivativeCoeffLemma_sin}
\end{lemma}

\begin{proof}
See Appendix \ref{proof_ADM:PartialDerivativeCoeffLemma_sin}. 
\end{proof}

\par From this, we have the following results. 

\begin{theorem}
\label{thm:ADM_order_n_sin}
Let $\{u_n(x,y,t)\}$, $\{v_n(x,y,t)\}$, $\{h_n(x,y,t)\}$ be the sequence of decomposed functions of $u$, $v$, and $h$, where their relationship is defined by \eqref{eq:ADM_iter_Linv}. If $D=0$ and the initial conditions $u_0(x,y)$, $v_0(x,y)$, $h_0(x,y)$ satisfy the properties defined in \eqref{PartialDerivIC_x_Cond_sin}  - \eqref{PartialDerivIC_mixedxy_Cond_sin},  then the solutions $u$, $v$, and $h$ have the following properties 

\begin{equation}
    \frac{\upartial u(x,y,t)}{\upartial x} = \frac{\upartial u(x,y,t)}{\upartial y} = 0,
\label{PartialCondition_u_sin}
\end{equation}

\begin{equation}
    \frac{\upartial v(x,y,t)}{\upartial x} = \frac{\upartial v(x,y,t)}{\upartial y} 
     = 0,
\label{PartialCondition_v_sin}
\end{equation}

\begin{equation}
    \frac{\upartial h(x,y,t)}{\upartial x}=\frac{\upartial h(x,y,0)}{\upartial x},\;\; \frac{\upartial h(x,y,t)}{\upartial y}= \frac{\upartial h(x,y,0)}{\upartial y},
\label{ParticalCondition_h_sinA}
\end{equation}

\noindent and 

\begin{equation}
  \frac{\upartial^2 h(x,y,t)}{\upartial x^2}=\frac{\upartial^2 h(x,y,t)}{\upartial x\upartial y}=\frac{\upartial^2 h(x,y,t)}{\upartial y^2}=0. 
\label{ParticalCondition_h_sinB}
\end{equation}

\noindent Additionally, $u$, $v$, and $h$ are reduced to the following forms 

\begin{equation}
    {u}(x,y,t)=\tilde{u}_0(t),
\label{uLinearCondition_sin}
\end{equation}
    
\begin{equation}
{v}(x,y,t)=\tilde{v}_0(t),
\label{vLinearCondition_sin}
\end{equation}

\noindent and 

\begin{equation}
  {h}(x,y,t)=\tilde{h}_0(t)+\tilde{h}_x x+\tilde{h}_yy,
\label{hLinearCondition_sin}
\end{equation}

\noindent where the coefficients $\tilde{u}_0(t)$, $\tilde{v}_0(t)$, and  $\tilde{h}_0(t)$ are time-dependent, while $\tilde{h}_x$ and $\tilde{h}_y$ are constants. Additionally, \eqref{uLinearCondition_sin} - \eqref{hLinearCondition_sin} satisfy the reduced system of equations

\begin{subequations}
\label{eq:ODE_sin}
\begin{align}
    \frac{\rm d}{{\rm d}t}\tilde{u}_0(t)=& -\bar{\tau} \tilde{u}_0(t)+\bar{f}\tilde{v}_0(t)-\frac{1}{F^2}\tilde{h}_x, \label{eq:ODEsinA} \\
    \frac{\rm d}{{\rm d}t}\tilde{v}_0(t)=& -\bar{f} \tilde{u}_0(t)-\bar{\tau} \tilde{v}_0(t)-\frac{1}{F^2}\tilde{h}_y,\;\; \label{eq:ODEsinB} \\
    \frac{\rm d}{{\rm d}t}\tilde{h}_0(t)=&-\tilde{h}_x\tilde{u}_0(t)  -\tilde{h}_y\tilde{v}_0(t), \label{eq:ODEsinC}
\end{align}
\end{subequations}

\noindent which yields the  solution 

\begin{align}
     \boldsymbol{\psi}(t)=&{\rm e}^{\boldsymbol{A}t}\boldsymbol{\psi}(0)+\int_{0}^t {\rm e}^{\boldsymbol{A}(t-\xi)}\boldsymbol{H}\; {\rm d}\xi,
     \label{eq:solution_decaying_inertial_geostrophic}
\end{align}

\noindent where 

\begin{align}
   \boldsymbol{\psi}(t)=\begin{bmatrix}
   \tilde{u}_0(t)\\
   \tilde{v}_0(t)\\
   \tilde{h}_0(t)
   \end{bmatrix}, \;\; \boldsymbol{\psi}(0)=\begin{bmatrix}
   \tilde{u}_0(0)\\
   \tilde{v}_0(0)\\
   \tilde{h}_0(0)
   \end{bmatrix},\;\;\boldsymbol{A}=\begin{bmatrix}
   -\bar{\tau} & \bar{f} & 0\\
   -\bar{f} & -\bar{\tau} & 0\\
   -\tilde{h}_x & -\tilde{h}_y & 0
   \end{bmatrix},\;\; \boldsymbol{H}=-\frac{1}{F^2}\begin{bmatrix}
   \tilde{h}_x\\
   \tilde{h}_y\\
   0
   \end{bmatrix},
   \nonumber
   \end{align}

\noindent and the components of $\left[{\rm e}^{\boldsymbol{A}t}\right]_{i,j}$ (for $i, j=1,2,3$) are given by
     
     \begin{equation}
         [{\rm e}^{\boldsymbol{A}t}]_{1,1}=[{\rm e}^{\boldsymbol{A}t}]_{2,2}={\rm e}^{-\bar{\tau} t }\cos(\bar{f}t),\nonumber
     \end{equation}
     
     \begin{equation}
        [{\rm e}^{\boldsymbol{A}t}]_{1,2}=-[{\rm e}^{\boldsymbol{A}t}]_{2,1}={\rm e}^{-\bar{\tau} t }\sin(\bar{f}t),  \nonumber 
     \end{equation}
     
     \begin{equation}
         [{\rm e}^{\boldsymbol{A}t}]_{1,3}=[{\rm e}^{\boldsymbol{A}t}]_{2,3}=0,\nonumber
     \end{equation}
     
     \begin{equation}
         [{\rm e}^{\boldsymbol{A}t}]_{3,1}=\frac{1}{\bar{f}^2 + \bar{\tau}^2} [(-\bar{f}\tilde{h}_y+\tilde{h}_x \bar{\tau}){\rm e}^{-\bar{\tau} t}\cos(\bar{f}t) + (-\bar{f}\tilde{h}_x-\tilde{h}_y\bar{\tau}){\rm e}^{-\bar{\tau} t}\sin(\bar{f}t) + \bar{f}\tilde{h}_y - \tilde{h}_x\bar{\tau}],\nonumber
     \end{equation}
     
     \begin{equation}
        [{\rm e}^{\boldsymbol{A}t}]_{3,2}=\frac{1}{\bar{f}^2 + \bar{\tau}^2} [ (\bar{f}\tilde{h}_x+\tilde{h}_y\bar{\tau}){\rm e}^{-\bar{\tau} t} \cos(\bar{f}t) +(-\bar{f}\tilde{h}_y+\tilde{h}_x\bar{\tau}){\rm e}^{-\bar{\tau} t}\sin(\bar{f}t) - \bar{f}\tilde{h}_x - \tilde{h}_y\bar{\tau} ],\nonumber
     \end{equation}
     
     \noindent and
     
     \begin{equation}
    [{\rm e}^{\boldsymbol{A}t}]_{3,3}=1.\nonumber
     \end{equation}

\end{theorem}

\begin{proof}
Applying Lemma \ref{ADM:PartialDerivativeCoeffLemma_sin} to each component in \eqref{eq:ADM_series} yields \eqref{PartialCondition_u_sin} - \eqref{ParticalCondition_h_sinB}. From \eqref{PartialCondition_u_sin}, we observe that 

\begin{equation}
   \frac{\upartial u(x,y,t)}{\upartial x} = 0 \, \, \, \, \mbox{yields} \, \, \, \, u(x,y,t)=C_1(y,t),
\label{integration_constant_C1C2_Result} \nonumber
\end{equation}

\noindent where the integration constant, $C_1(y,t)$, is independent of $x$. Similarly, we have 

\begin{equation}
   \frac{\upartial u(x,y,t)}{ \upartial y} = 0 \, \, \, \, \mbox{yields} \, \, \, \, C_1(y,t) = \tilde{u}_0(t)
\label{integration_constant_2_Result}\nonumber
\end{equation}

\noindent and thus \eqref{uLinearCondition_sin} is achieved. Similar arguments can be made to achieve \eqref{vLinearCondition_sin} and \eqref{hLinearCondition_sin}, respectively. The reduced equation in \eqref{eq:ODE_sin} is obtained by substituting \eqref{uLinearCondition_sin} - \eqref{hLinearCondition_sin} into \eqref{eq:SWE_non_dimen} which can further be expressed as

\begin{align}
     \frac{\rm d}{{\rm d}t}\boldsymbol{\psi}(t)=&\boldsymbol{A}\boldsymbol{\psi}(t)+\boldsymbol{H},
     \label{eq:eqaution_decaying_inertial_geostrophic}
\end{align} 

\noindent where $\boldsymbol{\psi}(0)$ is the vector containing the initial conditions. Solving \eqref{eq:eqaution_decaying_inertial_geostrophic} yields \eqref{eq:solution_decaying_inertial_geostrophic}. 
\end{proof}

\begin{corollary}
Let the initial behaviour over a flat bottom topography $D=0$ with constant normalised Coriolis force $\bar{f}\neq 0$ and friction coefficient $\bar{\tau}>0$ be defined as $u_0(x,y)=\tilde{u}_0(0)$, $v_0(x,y)=\tilde{v}_0(0)$, and $h_0(x,y)=\tilde{h}_0(0)+\tilde{h}_x x+\tilde{h}_y y $. The long-term behaviour is expressed as   

\begin{align}
    u(x,y,t)=&\frac{-\bar{f}\tilde{h}_y-\bar{\tau} \tilde{h}_x}{F^2(\bar{f}^2+\bar{\tau}^2)},\label{eq:inertial_geostrophic_inf_time_u}\\
    v(x,y,t)=&\frac{\bar{f}\tilde{h}_x-\bar{\tau} \tilde{h}_y}{F^2(\bar{f}^2+\bar{\tau}^2)},\label{eq:inertial_geostrophic_inf_time_v}\\
    h(x,y,t)=&\frac{\bar{\tau} (\tilde{h}_x^2+\tilde{h}_y^2)t}{F^2(\bar{f}^2+\bar{\tau}^2)}+\tilde{h}_0(0)+\tilde{h}_x x+\tilde{h}_y y.
    \label{eq:inertial_geostrophic_inf_time_h}
\end{align}
\label{cor:linear_growing}

\end{corollary}

\begin{proof}
See Appendix \ref{proof_cor:linear_growing}. 
\end{proof}

Theorem \ref{thm:ADM_order_n_sin} and Corollary \ref{cor:linear_growing} show the correlation between inertial geostrophic flows, in terms of their oscillatory behaviour, the constant Coriolis force $\bar{f}$, and the constant bottom friction $\bar{\tau}$. Theorem \ref{thm:ADM_order_n_sin} shows that the inertial oscillations depend on the constant Coriolis force $\bar{f} \ne 0$, which is indicated by the terms $\sin(\bar{f}t)$ and $\cos(\bar{f}t)$. However, the damping behaviour is directly proportional to the constant bottom friction as indicated by ${\rm e}^{-\bar{\tau} t}$ multiplied by each of these terms. Corollary \ref{cor:linear_growing} describes the limiting behaviour of these flows, where we note that the horizontal velocities reach constant values that are proportional to linear combinations of the free surface gradients in the $x$ and $y$ directions, the constant Coriolis force, the constant friction coefficient, and the Froude number as $t\rightarrow \infty$. Additionally, we observe that the limiting behavior of the free surface height grows linearly with respect to time at a rate that is proportional to $\left[\bar{\tau} (\tilde{h}_x^2+\tilde{h}_y^2)\right]/\left[F^2(\bar{f}^2+\bar{\tau}^2)\right]$. Furthermore, when $\bar{f}=0$ (and $\tilde{h}_x^2+\tilde{h}_y^2 \ne 0$) this rate grows proportionally to $1/\bar{\tau}$. 

\subsection{Anticyclonic vortices with a finite escape time}
\label{subsec:new_growing_anticyclonic_vortices}

For these types of flows, our analysis considers the following initial conditions 

\begin{itemize}
    \item Condition V
    \begin{equation}
        u_0(x,y)=\bar{f}y-\bar{\tau} x,\;\;v_0(x,y)=0,\;\;h_0(x,y)=\tilde{h}_0(0),
        \label{eq:IC_cond_V}
    \end{equation}
    
    \item Condition VI
    \begin{equation}
        u_0(x,y)=0,\;\;v_0(x,y)=-\bar{f}x-\bar{\tau} y,\;\;h_0(x,y)=\tilde{h}_0(0),
        \label{eq:IC_cond_VI}
    \end{equation}

    \item Condition VII
    \begin{equation}
        u_0(x,y)=\bar{f}y -\bar{\tau} x,\;\;v_0(x,y)=-\bar{f}x-\bar{\tau} y,\;\;h_0(x,y)=\tilde{h}_0(0),
        \label{eq:IC_cond_VII}
    \end{equation}
\end{itemize}

\noindent where $\tilde{h}_0$ is the constant free surface height. These describe anticyclonic vortices because the initial vorticity is proportional to the negative of the constant Coriolis parameter. The behaviour of the initial conditions \eqref{eq:IC_cond_V} - \eqref{eq:IC_cond_VI} affect the decomposition of the decomposed functions of $u$, $v$, and $h$ as presented in the following results.

\begin{lemma}
Let $\{u_n(x,y,t)\}$, $\{v_n(x,y,t)\}$, $\{h_n(x,y,t)\}$ be the sequence of decomposed functions of $u$, $v$, and $h$, where their relationship is defined by \eqref{eq:ADM_iter_Linv} (for $n \in \mathbb{N} $) given a flat bottom topography $D=0$. If the initial conditions $u_0(x,y)$, $v_0(x,y)$, $h_0(x,y)$ are defined such that 

\begin{equation}
    u_0(x,y)=\bar{f}y-\bar{\tau} x,
\label{PartialDerivIC_x_Cond_tan_u_fy}
\end{equation}

\begin{equation}
    \frac{\upartial^2 v_0(x,y)}{\upartial x^2}=\frac{\upartial h_0(x,y)}{\upartial x}=0,
\label{PartialDerivIC_x_Cond_tan_u}
\end{equation}

\begin{equation}
    \frac{\upartial^2 v_0(x,y)}{\upartial y^2}=\frac{\upartial h_0(x,y)}{\upartial y}=0,
\label{PartialDerivIC_y_Cond_tan_u}
\end{equation}

\noindent and 

\begin{equation}
    \frac{\upartial^2 v_0(x,y)}{\upartial x\upartial y}=0.
\label{PartialDerivIC_mixedxy_Cond_tan_u}
\end{equation}

\noindent Then the higher order components $u_n(x,y,t)$, $v_n(x,y,t)$, $h_n(x,y,t)$, (for $n\in \mathbb{N}^{+}$) satisfy

\begin{equation}
u_n(x,y,t)=0,
\label{PartialDeriv_x_Cond_tan_u_fy}
\end{equation}

\begin{equation}
    \frac{\upartial^2 v_n(x,y,t)}{\upartial x^2}=\frac{\upartial h_n(x,y,t)}{\upartial x}=0,
\label{PartialDeriv_x_Cond_tan_u}
\end{equation}

\begin{equation}
    \frac{\upartial^2 v_n(x,y,t)}{\upartial y^2}= \frac{\upartial h_n(x,y,t)}{\upartial y}=0,
\label{PartialDeriv_y_Cond_tan_u}
\end{equation}

\noindent and 

\begin{equation}
        \frac{\upartial^2 v_n(x,y,t)}{\upartial x\upartial y}=0.
\label{PartialDeriv_mixedxy_Cond_tan_u}
\end{equation}
\label{ADM:PartialDerivativeCoeffLemma_tan_u}
\end{lemma}

\begin{proof}
See Appendix \ref{proof_ADM:PartialDerivativeCoeffLemma_tan_u}. 
\end{proof}

\begin{theorem}
\label{thm:ADM_order_n_tan_u}
Let $\{u_n(x,y,t)\}$, $\{v_n(x,y,t)\}$, $\{h_n(x,y,t)\}$ be the sequence of decomposed functions of $u$, $v$, and $h$, where their relationship is defined by \eqref{eq:ADM_iter_Linv} (for $n \in \mathbb{N}$) given a flat bottom topography $D=0$. If the initial conditions $u_0(x,y)$, $v_0(x,y)$, $h_0(x,y)$ are defined as \eqref{PartialDerivIC_x_Cond_tan_u_fy}  - \eqref{PartialDerivIC_mixedxy_Cond_tan_u},  then the solutions of $u$, $v$, and $h$ have the property where 

\begin{equation}
    u(x,y,t)=\bar{f}y-\bar{\tau} x,
\label{PartialCondition_u_tan_u_fy}
\end{equation}

\begin{equation}
    \frac{\upartial^2 v(x,y,t)}{\upartial x^2}=\frac{\upartial^2 v(x,y,t)}{\upartial y^2}=\frac{\upartial^2 v(x,y,t)}{\upartial x\upartial y}=0,
\label{PartialCondition_v_tan_u}
\end{equation}

\noindent and 

\begin{equation}
    \frac{\upartial h(x,y,t)}{\upartial x}=\frac{\upartial h(x,y,t)}{\upartial y}=0.
\label{PartialCondition_h_tan_u}
\end{equation}

\noindent Consequently, these solutions can be expressed as 
\begin{equation}
    {u}(x,y,t)=\bar{f}y-\bar{\tau} x,
\label{uLinearCondition_tan_u}
\end{equation}
    
\begin{equation}
{v}(x,y,t)=\tilde{v}_0(t)+\tilde{v}_x(t)x+\tilde{v}_y(t)y,
\label{vLinearCondition_tan_u}
\end{equation}

\noindent and 

\begin{equation}
  {h}(x,y,t)=\tilde{h}_0(t),
\label{hLinearCondition_tan_u}
\end{equation}

\noindent where the coefficients  $\tilde{v}_0(t)$, $\tilde{v}_x(t)$, $\tilde{v}_y(t)$, and $\tilde{h}_0(t)$ are time-dependent. These coefficients satisfy
\begin{subequations}
     \label{eq:ODE_tan_u}
\begin{align}
    \frac{\rm d}{{\rm d}t}\tilde{v}_0(t)=&-\tilde{v}_0(t) \tilde{v}_y(t)-\bar{\tau} \tilde{v}_0(t),\\
    \frac{\rm d}{{\rm d}t}\tilde{v}_x(t)=&\bar{f}\bar{\tau} -\tilde{v}_x(t)\tilde{v}_y(t),\\
    \frac{\rm d}{{\rm d}t}\tilde{v}_y(t)=&-\bar{f}\tilde{v}_x(t)-\tilde{v}_y(t)^2-\bar{f}^2-\bar{\tau} \tilde{v}_y(t),\;\\
    \frac{\rm d}{{\rm d}t}\tilde{h}_0(t)=&-\tilde{h}_0(t)\tilde{v}_y(t)+\bar{\tau} \tilde{h}_0(t).
\end{align}
\end{subequations}

\end{theorem}

\begin{proof}
Applying Lemma \ref{ADM:PartialDerivativeCoeffLemma_tan_u} to each component in \eqref{eq:ADM_series} yields \eqref{PartialCondition_u_tan_u_fy} - \eqref{PartialCondition_h_tan_u}. From \eqref{PartialCondition_v_tan_u}, we observe that 

\begin{equation}
   \frac{\upartial^2 v(x,y,t)}{\upartial x^2} = 0 \, \, \, \, \mbox{yields} \, \, \, \, v(x,y,t)=C_1(y,t) x+C_2(y,t),
\label{integration_constant_C1C2_Result_tan_u} \nonumber
\end{equation}

\noindent where the integration constants, $C_1(y,t)$ and $C_2(y,t)$, are independent of $x$. Similarly, we have 

\begin{equation}
   \frac{\upartial^2 v(x,y,t)}{\upartial x \upartial y} = 0 \, \, \, \, \mbox{yields} \, \, \, \, C_1(y,t) = \tilde{v}_x(t)
\label{integration_constant_2_Result_tan_u}\nonumber
\end{equation}

\noindent and 

\begin{equation}
\frac{\upartial^2 v(x,y,t)}{\upartial y^2} = 0 \, \, \, \, \mbox{yields} \, \, \, \, C_2(y,t)=\tilde{v}_y(t) y+\tilde{v}_0(t),
\label{integration_constantC2_Result_tan_u}\nonumber
\end{equation}

\noindent and thus \eqref{vLinearCondition_tan_u} is achieved. Similar arguments can be made to achieve \eqref{uLinearCondition_tan_u} and \eqref{hLinearCondition_tan_u}, respectively. The reduced equation in \eqref{eq:ODE_tan_u} is obtained by substituting \eqref{uLinearCondition_tan_u} - \eqref{hLinearCondition_tan_u} into \eqref{eq:SWE_non_dimen}.
\end{proof}

\begin{lemma}
Let $\{u_n(x,y,t)\}$, $\{v_n(x,y,t)\}$, $\{h_n(x,y,t)\}$ be the sequence of decomposed functions of $u$, $v$, and $h$, where their relationship is defined by \eqref{eq:ADM_iter_Linv} (for $n \in \mathbb{N}$) given a flat bottom topography $D=0$. If the initial conditions $u_0(x,y)$, $v_0(x,y)$, $h_0(x,y)$ are defined as

\begin{align}
    v_0(x,y)=-\bar{f}x-\bar{\tau} y,
    \label{PartialDerivIC_x_tan_v_fx}
\end{align}

\begin{align}
  \frac{\upartial^2 u_0(x,y)}{\upartial x^2}=\frac{\upartial h_0(x,y)}{\upartial x}=0,
  \label{PartialDerivIC_x_Cond_tan_v}
\end{align}

\begin{align}
    \frac{\upartial^2 u_0(x,y)}{\upartial y^2}=\frac{\upartial h_0(x,y)}{\upartial y}=0,
    \label{PartialDerivIC_y_Cond_tan_v}
\end{align}

\noindent and

\begin{align}
    \frac{\upartial^2 u_0(x,y)}{\upartial x\upartial y}=0.
    \label{PartialDerivIC_mixedxy_Cond_tan_v}
\end{align}

\noindent Then the higher order components $u_n(x,y,t)$, $v_n(x,y,t)$, $h_n(x,y,t)$, (for $n\in \mathbb{N}^{+}$) satisfy the property

\begin{align}
    v_n(x,y,t)=0,
    \label{PartialDeriv_x_Cond_tan_v_fx}
\end{align}

\begin{align}
  \frac{\upartial^2 u_n(x,y,t)}{\upartial x^2}=\frac{\upartial h_n(x,y,t)}{\upartial x}=0,
  \label{PartialDeriv_x_Cond_tan_v}
\end{align}

\begin{align}
    \frac{\upartial^2 u_n(x,y,t)}{\upartial y^2}=\frac{\upartial h_n(x,y,t)}{\upartial y}=0,
    \label{PartialDeriv_y_Cond_tan_v}
\end{align}

\noindent and

\begin{align}
    \frac{\upartial^2 u_n(x,y,t)}{\upartial x \upartial y}=0.
    \label{PartialDeriv_mixedxy_Cond_tan_v}
\end{align}

\label{ADM:PartialDerivativeCoeffLemma_tan_v}
\end{lemma}

\begin{proof}
See Appendix \ref{proof_ADM:PartialDerivativeCoeffLemma_tan_v}. 
\end{proof}

\begin{theorem}
\label{thm:ADM_order_n_tan_v}
Let $\{u_n(x,y,t)\}$, $\{v_n(x,y,t)\}$, $\{h_n(x,y,t)\}$ be the sequence of decomposed functions of $u$, $v$, and $h$, where their relationship is defined by \eqref{eq:ADM_iter_Linv} (for $n \in \mathbb{N} $) given a flat bottom topography $D=0$. If the initial conditions $u_0(x,y)$, $v_0(x,y)$, $h_0(x,y)$ are defined as \eqref{PartialDerivIC_x_tan_v_fx} -   \eqref{PartialDerivIC_mixedxy_Cond_tan_v},  then the solutions of $u$, $v$, and $h$ have the property where 

\begin{equation}
    \frac{\upartial^2 u(x,y,t)}{\upartial x^2} = \frac{\upartial^2 u(x,y,t)}{\upartial y^2} 
    = \frac{\upartial^2 u(x,y,t)}{\upartial x \upartial y} = 0,
\label{PartialCondition_u_tan_v}
\end{equation}

\begin{equation}
    v(x,y,t)=-\bar{f}x-\bar{\tau} y,
    \label{PartialCondition_v_tan_v}
\end{equation}

\noindent and 

\begin{equation}
    \frac{\upartial h(x,y,t)}{\upartial x}
    = \frac{\upartial h(x,y,t)}{\upartial y} = 0.
\label{PartialCondition_h_tan_v}
\end{equation}

\noindent Consequently, these solutions can be expressed as 
\begin{equation}
    {u}(x,y,t)=\tilde{u}_0(t)+\tilde{u}_x(t)x+\tilde{u}_y(t)y,
\label{uLinearCondition_tan_v}
\end{equation}
    
\begin{equation}
    {v}(x,y,t)=-\bar{f}x-\bar{\tau} y,
\label{vLinearCondition_tan_v}
\end{equation}

\noindent and 

\begin{equation}
  {h}(x,y,t)=\tilde{h}_0(t),  
\label{hLinearCondition_tan_v}
\end{equation}

\noindent where the coefficients $\tilde{u}_0(t)$, $\tilde{u}_x(t)$, $\tilde{u}_y(t)$, and $\tilde{h}_0(t)$,  are time-dependent. These coefficients satisfy
\begin{subequations}
     \label{eq:ODE_tan_v}
\begin{align}
    \frac{\rm d}{{\rm d}t}\tilde{u}_0(t)=&-\tilde{u}_0(t) \tilde{u}_x(t)-\bar{\tau} \tilde{u}_0(t), \\
    \frac{\rm d}{{\rm d}t}\tilde{u}_x(t)=&-\tilde{u}_x(t)^2+\bar{f}\tilde{u}_y(t)-\bar{f}^2-\bar{\tau} \tilde{u}_x(t),\\
    \frac{\rm d}{{\rm d}t}\tilde{u}_y(t)=&- \tilde{u}_y(t) \tilde{u}_x(t)-\bar{\tau} \bar{f} ,\\
    \frac{\rm d}{{\rm d}t}\tilde{h}_0(t)=&-\tilde{h}_0(t)\tilde{u}_x(t).
\end{align}
\end{subequations}
\end{theorem}

\begin{proof}
Applying Lemma \ref{ADM:PartialDerivativeCoeffLemma_tan_v} to each component in \eqref{eq:ADM_series} yields \eqref{PartialCondition_u_tan_v} - \eqref{PartialCondition_h_tan_v}. From \eqref{PartialCondition_u_tan_v}, we observe that 

\begin{equation}
   \frac{\upartial^2 u(x,y,t)}{\upartial x^2} = 0 \, \, \, \, \mbox{yields} \, \, \, \, u(x,y,t)=C_1(y,t) x+C_2(y,t),
\label{integration_constant_C1C2_Result}\nonumber
\end{equation}

\noindent where the integration constants, $C_1(y,t)$ and $C_2(y,t)$, are independent of $x$. Similarly, we have 

\begin{equation}
   \frac{\upartial^2 u(x,y,t)}{\upartial x \upartial y} = 0 \, \, \, \, \mbox{yields} \, \, \, \, C_1(y,t) = \tilde{u}_x(t)
\label{integration_constant_2_Result} \nonumber
\end{equation}

\noindent and 

\begin{equation}
\frac{\upartial^2 u(x,y,t)}{\upartial y^2} = 0 \, \, \, \, \mbox{yields} \, \, \, \, C_2(y,t)=\tilde{u}_y(t) y+\tilde{u}_0(t),
\label{integration_constantC2_Result} \nonumber
\end{equation}

\noindent and thus \eqref{uLinearCondition_tan_v} is achieved. Similar arguments can be made to achieve \eqref{vLinearCondition_tan_v} and \eqref{hLinearCondition_tan_v}. The reduced equation in \eqref{eq:ODE_tan_v} is obtained by substituting \eqref{uLinearCondition_tan_v} - \eqref{hLinearCondition_tan_v} into \eqref{eq:SWE_non_dimen}. 
\end{proof}
     
\begin{theorem}
\label{thm:anticyclonic_finite_escape_u}
For any flows over flat bottom topography $(D=0)$ with constant Coriolis parameter $\bar{f}\neq 0$ and initial constant free surface height $\tilde{h}_0(0)$, the solution $u$, $v$, and $h$ with respect to their initial conditions are defined as follows. If the initial behaviour is defined by

 \begin{equation}
     u_0(x,y)=\bar{f}\,y-\bar{\tau} \,x,\;\;
     v_0(x,y)=0, \;\;
     h_0(x,y)=\tilde{h}_{0}(0),\label{eq:case51_initial}
     \end{equation}

\noindent then

  \begin{align}
     u(x,y,t)=&\bar{f}\,y-\bar{\tau} \,x,
     \nonumber\\
     v(x,y,t)=&-\bar{f}\,y\,\mathrm{tan}\left(\bar{f}\,t\right)+\bar{\tau} \,x\,\mathrm{tan}\left(\bar{f}\,t\right),
     \nonumber\\
     h(x,y,t)=& \tilde{h}_{0}(0)\mathrm{e}^{\bar{\tau}\, t }\sec \left(\bar{f}\,t\right).
     \label{eq:case51_exact_c}
     \end{align}
     
     \noindent Furthermore, this solution describes anticyclonic vortices with a finite escape time that is based on initial zonal velocity being represented as $u_0(x,y)=\bar{f}y-\bar{\tau} x$.
     
     \end{theorem}

\begin{proof}
This is obtained by solving \eqref{eq:ODE_tan_u} in Theorem \ref{thm:ADM_order_n_tan_u}.
\end{proof}
    
\begin{theorem}
\label{thm:anticyclonic_finite_escape_v}
For any flows over flat bottom topography $(D=0)$ and constant Coriolis parameter $\bar{f}\neq 0$ and initial constant free surface height $\tilde{h}_0(0)$, the solution $u$, $v$, and $h$ with respect to their initial conditions are defined as follows. If the initial behaviour is defined by
     \begin{equation}
     u_0(x,y)=0,\;\;
     v_0(x,y)=-\bar{f}\,x-\bar{\tau} \,y,\;\;
     h_0(x,y)=\tilde{h}_{0}(0),\label{eq:case52_initial}
     \end{equation}

     \noindent then

     \begin{align}
     u(x,y,t)=&-\bar{f}\,x\,\mathrm{tan}\left(\bar{f}\,t\right)-\bar{\tau} \,y\,\mathrm{tan}\left(\bar{f}\,t\right),\nonumber\\
     v(x,y,t)=&-\bar{f}\,x-\bar{\tau} \,y,\nonumber\\
      h(x,y,t)=& \tilde{h}_{0}(0)\mathrm{e}^{\bar{\tau}\, t }\sec \left(\bar{f}\,t\right).
     \label{eq:case52_exact_c}
     \end{align}

    \noindent Furthermore, this solution describes anticyclonic vortices with a finite escape time that is based on initial meridional velocity being represented as $v_0(x,y)=-\bar{f}x-\bar{\tau} y$.

\end{theorem}

\begin{proof}
This is obtained by solving \eqref{eq:ODE_tan_v} in Theorem \ref{thm:ADM_order_n_tan_v}.
\end{proof}

\begin{theorem}
\label{thm:theorem_exp_grow}
Let $\{u_n(x,y,t)\}$, $\{v_n(x,y,t)\}$, $\{h_n(x,y,t)\}$ be the sequence of decomposed functions of $u$, $v$, and $h$, where their relationship is defined by \eqref{eq:ADM_iter_Linv} (for $n \in \mathbb{N}$) given a flat bottom topography $D=0$. If the initial conditions $u_0(x,y)$ and $v_0(x,y)$ are defined by \eqref{eq:IC_cond_VII} where the initial free surface behaviour $h_0(x,y)$ has the following property 

\begin{equation}
    \frac{\upartial h_0(x,y)}{\upartial x}=\frac{\upartial h_0(x,y)}{\upartial y}=0
    \label{IC_h_exp_grow},
\end{equation}

\noindent then the solutions of $u$, $v$, and $h$ have the property where 

\begin{equation}
    u(x,y,t)=\bar{f}y-\bar{\tau} x,
\label{PartialCondition_u_exp_grow}
\end{equation}

\begin{equation}
v(x,y,t)=-\bar{f}x-\bar{\tau} y,
\label{PartialCondition_v_exp_grow}
\end{equation}

\noindent and 

\begin{equation}
    \frac{\upartial h(x,y,t)}{\upartial x}=\frac{\upartial h(x,y,t)}{\upartial y}=0.
\label{PartialCondition_h_exp_grow}
\end{equation}

\noindent Consequently, these solutions can be expressed as 
\begin{equation}
    {u}(x,y,t)=\bar{f}y-\bar{\tau} x,
\label{uLinearCondition_linear_u_v}
\end{equation}
    
\begin{equation}
{v}(x,y,t)=-\bar{f}x-\bar{\tau} y,
\label{vLinearCondition_linear_u_v}
\end{equation}

\noindent and 

\begin{equation}
  {h}(x,y,t)=\tilde{h}_0(t),
\label{hLinearCondition_linear_u_v}
\end{equation}

\noindent where the coefficients  $\tilde{h}_0(t)$ are time-dependent and satisfy
\begin{equation}
    \tilde{h}_0(t)= {\rm e}^{2\bar{\tau}\, t}\tilde{h}_0(0).
     \label{ODE_linear_h_exp_grow_solution}
\end{equation}

\noindent Furthermore, this solution describes anticyclonic vortices with exponentially growing free surface height that is based on the initial zonal velocity $u_0(x,y)=\bar{f}y-\bar{\tau} x$ and the initial meridional velocity $v_0(x,y)=-\bar{f}x-\bar{\tau} y$.

\end{theorem}

\begin{proof}   
The initial conditions $u_0(x,y)$, $v_0(x,y)$, and $h_0(x,y)$ satisfy  \eqref{PartialDerivIC_x_Cond_tan_u_fy}  - \eqref{PartialDerivIC_mixedxy_Cond_tan_u} as well as \eqref{PartialDerivIC_x_tan_v_fx} - \eqref{PartialDerivIC_mixedxy_Cond_tan_v} and thus Lemmas \ref{ADM:PartialDerivativeCoeffLemma_tan_u} and \ref{ADM:PartialDerivativeCoeffLemma_tan_v} satisfied. From Theorems \ref{thm:ADM_order_n_tan_u} and \ref{thm:ADM_order_n_tan_v} we observe that

\begin{equation}
     u(x,y,t)=\bar{f}y-\bar{\tau} x, \nonumber
\end{equation}

\begin{equation}
    v(x,y,t)=-\bar{f}x-\bar{\tau} y, \nonumber
\end{equation}

\noindent and

\begin{equation}
    \frac{\upartial h(x,y,t)}{\upartial x}=\frac{\upartial h(x,y,t)}{\upartial y}=0. \nonumber
\end{equation}

\noindent Consequently, we can write

\begin{equation}
    h(x,y,t)=\tilde{h}_0(t), \nonumber
\end{equation}

\noindent where the coefficient $\tilde{h}_0(t)$ is time-dependent and satisfies

\begin{equation}
   \frac{\rm d}{{\rm d}t}\tilde{h}_0(t)=2\bar{\tau} \tilde{h}_0(t).
   \label{ODE_linear_h_exp_grow_equation}
\end{equation}

\noindent Solving \eqref{ODE_linear_h_exp_grow_equation} gives \eqref{ODE_linear_h_exp_grow_solution}. 
\end{proof}

\par Theorems \ref{thm:anticyclonic_finite_escape_u} and \ref{thm:anticyclonic_finite_escape_v} illustrate the effects of anticyclonic vortices with finite escape times since these solutions are valid for $t \in \left[0,  \pi/\left(2 \bar{f}\right) \right)$, which also include the effects of constant bottom friction in the velocity descriptions. These results also show the effect of constant bottom friction on the free surface height, where this phenomenon grows exponentially at a rate that is directly proportional to the constant friction coefficient $\bar{\tau}$. Theorem \ref{thm:theorem_exp_grow} provides further generalisation via considering the initial zonal and meridional velocity effects that are a linear combinations of constant Coriolis force and constant bottom friction. These still describe anticyclonic vortices due to the fact that the initial vorticity is negatively proportional to the Coriolis force. Hence, the velocity phenomena $u$ and $v$ are consistent with the initial behaviour; however, the free surface height grows exponentially faster in this case where the growth rate is directly proportional to twice the constant bottom friction coefficient. 

\par These results, corresponding to Conditions I - VII, are indicative of dissipation-induced instability in shallow water waves with respect to constant
bottom friction. For instance, examining Theorem \ref{thm:ADM_order_n_sin} we note that for the case of negligible bottom friction ($\bar{\tau} = 0$) this behaviour is purely oscillatory where the inertial oscillation frequency depends on the constant Coriolis force $\bar{f}$ \citep[Theorem 3.3 and Corollary 3.4]{liu2021ADMSWE}. Additionally, for the phenomena of anticyclonic vortices with finite escape times, we note that the initial dynamics have a profound impact on this type of stability behaviour which extends the analysis and results of \citep[Theorem 3.3 through Theorem 3.10]{liu2021ADMSWE}. Overall, these results illustrate the direct interplay between $\bar{\tau}$ and $\bar{f}$ on this type of phenomena where they are present in both short and long-term behaviour that advance previous results that only highlight the impact of these effects on this type of stability behaviour \citep{bloch1994dissipation,krechetnikov2007dissipation,krechetnikov2009dissipation}. These results considering the nonlinear effect also improve the analysis of \citet{magdalena2022effect}.

\section{Numerical validation and results}
\label{sec:numerical}

Numerical validation is done by comparing the respective convergence and accuracy of the partial sums ($S_N(u)$, $S_N(v)$, and $S_N(h)$) against the governing equations \eqref{eq:SWE_non_dimen}, the exact solutions ($u$, $v$, and $h$), and the numerical solutions ($\hat{u}$, $\hat{v}$, and $\hat{h}$) for Conditions 1 - VII via the relative integral squared error defined as

\begin{equation}
	E(N) = \frac{\int_{-L_x}^{L_x}\int_{-L_y}^{L_y}\int_0^{T} e(N;x,y,t) \,dt\,dx\,dy}{\int_{-L_x}^{L_x}\int_{-L_y}^{L_y}\int_0^{T}(u^2+v^2+h^2) \,dt\,dx\,dy},
    \label{eq:E_int}
\end{equation}

\begin{table}
    \centering
    \caption{Summary of the relevant parameters used to validate Conditions I-VII. We set $F=1$, $\bar{f}=0.5$, $\bar{\tau}=1$, $D_0=0$ for all Conditions.  }
    \begin{tabular}{cccccc}
        Cond. No. &  $u_0(x,y)$ & $v_0(x,y)$ & $h_0(x,y)$ & Value  &  Solutions \\
        \hline
        I &  $-\tilde{h}_y/(F^2\bar{f})$ & $\tilde{h}_x/(F^2 \bar{f})$ & $\tilde{h}_x x+\tilde{h}_y y$ & $\tilde{h}_x=\tilde{h}_y=10^{-4}$ & Theorem \ref{thm:ADM_order_n_sin}\\
        II &  0& 0 & $\tilde{h}_x x+ \tilde{h}_y y$ & $\tilde{h}_x=\tilde{h}_y=10^{-4}$ & Theorem \ref{thm:ADM_order_n_sin}\\
        III &  0& 0 & $\tilde{h}_x x$ & $\tilde{h}_x=10^{-4}$ & Theorem \ref{thm:ADM_order_n_sin}\\
        IV &  0& 0 &  $\tilde{h}_y y$ & $\tilde{h}_y=10^{-4}$ & Theorem \ref{thm:ADM_order_n_sin}\\
        V &   $\bar{f}y-\bar{\tau} x$ & 0 & $\tilde{h}_0(0)$ & $\tilde{h}_0(0)=10^{-4}$ & Theorem \ref{thm:anticyclonic_finite_escape_u}\\
        VI &  0 & $-\bar{f}x-\bar{\tau} y$ &  $\tilde{h}_0(0)$ & $\tilde{h}_0(0)=10^{-4}$ & Theorem \ref{thm:anticyclonic_finite_escape_v}\\
        VII &  $\bar{f}y-\bar{\tau} x$ & $-\bar{f}x-\bar{\tau} y$ &  $\tilde{h}_0(0)$ &  $\tilde{h}_0(0)=10^{-4}$ & Theorem \ref{thm:theorem_exp_grow}\\
        \hline
    \end{tabular}
    \label{tab:parameters_summary_friction}
\end{table}

\noindent where $L_x=1$, $L_y=1$, and $T=1$. The convergence $E_c(N)$ is measured by evaluating \eqref{eq:E_int} with

\begin{align}
    &e(N;x,y,t) = \nonumber\\
    &\Bigg(\frac{\upartial S_N(u)}{\upartial t}+S_N(u)\frac{\upartial S_N(u)}{\upartial x}+S_N(v)\frac{\upartial S_N(u)}{\upartial y}+\frac{1}{F^2}\frac{\upartial S_N(h)}{\upartial x}-\bar{f}S_N(v)+\bar{\tau} S_N(u)\Bigg)^2\nonumber\\
    &+\Bigg(\frac{\upartial S_N(v)}{\upartial t}+S_N(u)\frac{\upartial S_N(v)}{\upartial x}+S_N(v)\frac{\upartial S_N(v)}{\upartial y}+\frac{1}{F^2}\frac{\upartial S_N(h)}{\upartial y}+\bar{f}S_N(u)+\bar{\tau} S_N(v)\Bigg)^2\nonumber\\
    &+\Bigg(\frac{\upartial S_N(h)}{\upartial t}+\frac{\upartial }{\upartial x}\Big\{S_N(u)[S_N(h)+D]\Big\}+\frac{\upartial }{\upartial y}\Big\{S_N(v)[S_N(h)+D]\Big\}\Bigg)^2.
    \label{eq:square_residue}
\end{align}
$E_{ex}(N)$ is the accuracy of the partial sums of $u$, $v$, and $h$ against the exact solutions which is measured via evaluating \eqref{eq:E_int} with
\begin{equation}
    e(N;x,y,t)=(S_N(u)-u)^2+(S_N(v)-v)^2+(S_N(h)-h)^2.
\end{equation}
$\hat{E}(N)$ is the accuracy of the numerical solutions against the partial sums of $u$, $v$, and $h$ which is measured via evaluating \eqref{eq:E_int} with
\begin{equation}
    e(N;x,y,t)=(S_N(u)-\hat{u})^2+(S_N(v)-\hat{v})^2+(S_N(h)-\hat{h})^2.
\end{equation}
$\hat{E}_{ex}$ is the accuracy between the numerical and exact solutions, which is measured via evaluating \eqref{eq:E_int} with
\begin{equation}
    e(N;x,y,t)=(u-\hat{u})^2+(v-\hat{v})^2+(h-\hat{h})^2.
\end{equation}

In all evaluations, we follow \citet{matskevich2019exact} where the numerical implementations were done using the large-particle method and $F=1$ represents the characteristic velocity as $U_0=\sqrt{gH}$. The integration operator in equation \eqref{eq:E_int} is discretised with spatial and temporal grid spacings of $\Delta x=0.1$, $\Delta y=0.1$, and $\Delta t=0.1$. Furthermore, Table \ref{tab:parameters_summary_friction} lists the relevant parameters used to validate Conditions 1-VII including constant parameters, initial conditions, and the derived exact solutions as described in Section \ref{sec:new_exact}.

\subsection{Results}

\begin{table}
    \centering
\caption[center]{Convergence trend ($N=2,4$, and 6) and accuracy ($N=6$) summary.}
    \begin{tabular}{ccccccc}
         & $E_c(N=2)$ & $E_c(N=4)$ & $E_c(N=6)$ & $E_{ex}(N=6)$ & $\hat{E}(N=6)$ & $\hat{E}_{ex}$ \\ 
         \hline
        Minimum & $3.3\times 10^{-3}$ & $6.5\times 10^{-5}$ & $1.1\times 10^{-6}$ & $1.4\times 10^{-12}$ & $4.9\times 10^{-6}$ & $4.5\times 10^{-6}$ \\
        Maximum & $4.0\times 10^{-1}$ & $3.0\times 10^{-2}$ & $4.3\times 10^{-4}$ & $2.6\times 10^{-8}$  & $9.0\times 10^{-3}$ & $9.0\times 10^{-3}$ \\
        \hline
    \end{tabular}
    \label{tab:residue_summary_friction}
\end{table}

Table \ref{tab:residue_summary_friction} summarises the convergence and accuracy results, where we observe the errors stabilise within ${\rm O}(10^{-6})$ and ${\rm O}(10^{-4})$ when  $N=6$. This suggests that the partial sums of no more than six terms, from Adomian decomposition methods, provide effective approximations. The accuracy of the partial sums further validates these assertions where we observe accuracies between ${\rm O}(10^{-6})$ and ${\rm O}(10^{-3})$. The exact solutions for Conditions I-VII are also promising, where we note that the deviations from the numerical approximations are also small. 

\begin{figure}[h!]
    \centering
    (a) $t=0$ \hspace{0.35\textwidth} (b) $t=10,000$
    \includegraphics[width=0.49\textwidth]{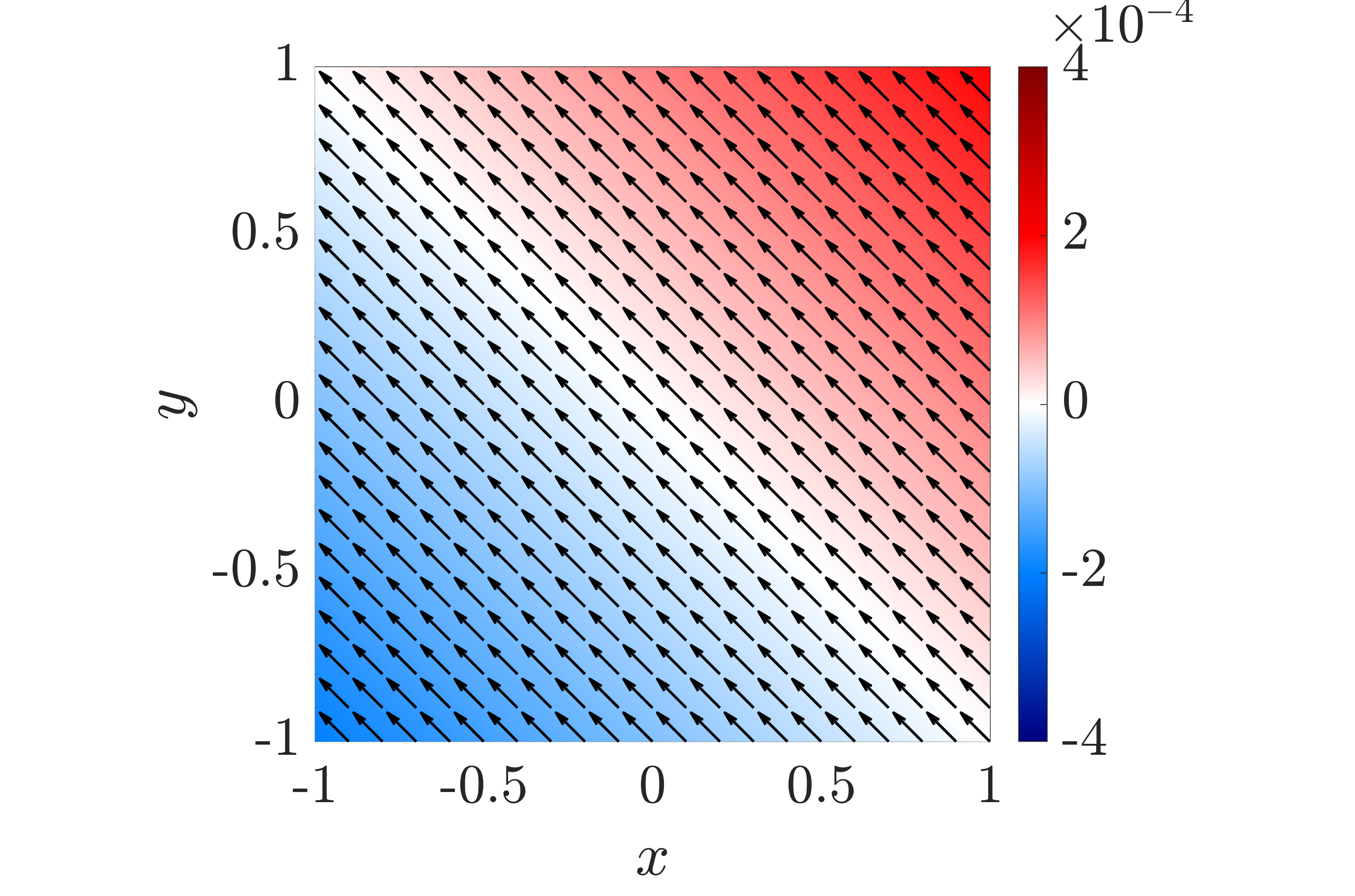}
    \includegraphics[width=0.49\textwidth]{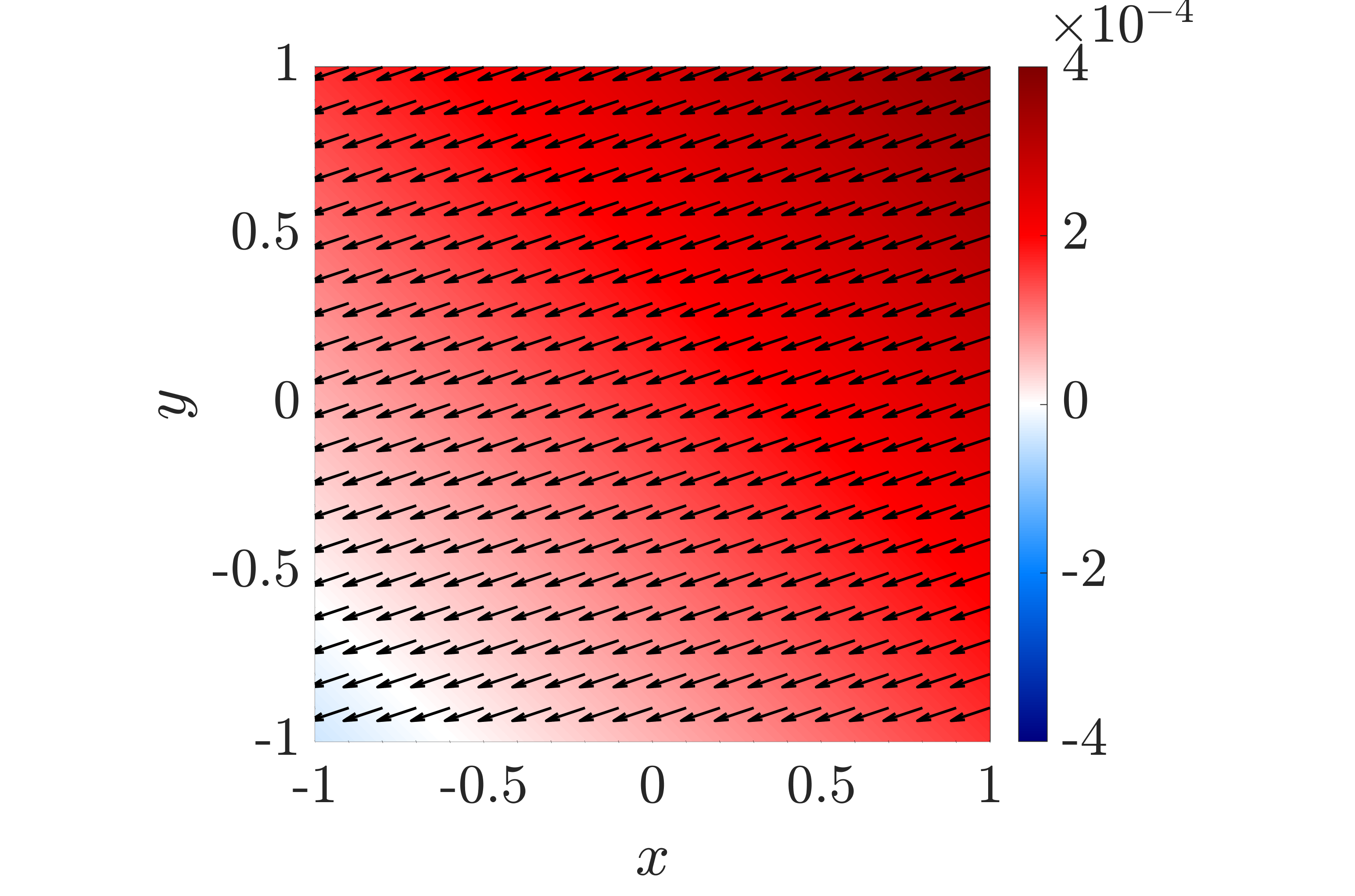}
    \caption{Velocity vector field $\boldsymbol{u}$ (arrow) and free surface height (contour) behaviour for Condition I: (a) initial condition at $t=0$ and (b) exact solution at $t=10,000$ based on Theorem \ref{thm:ADM_order_n_sin}. Parameters used include $F=1$, $\bar{f}=0.5$, $\bar{\tau}=1$, $D_0=0$, and $\tilde{h}_x=\tilde{h}_y=10^{-4}$ (Colour online).}
    \label{fig:inertial_geostrophic_t_10000}
\end{figure}

\begin{figure}[h!]
    \centering
    (a) $t=0$ \hspace{0.4\textwidth} (b) $t=1$ 
    \includegraphics[width=0.49\textwidth]{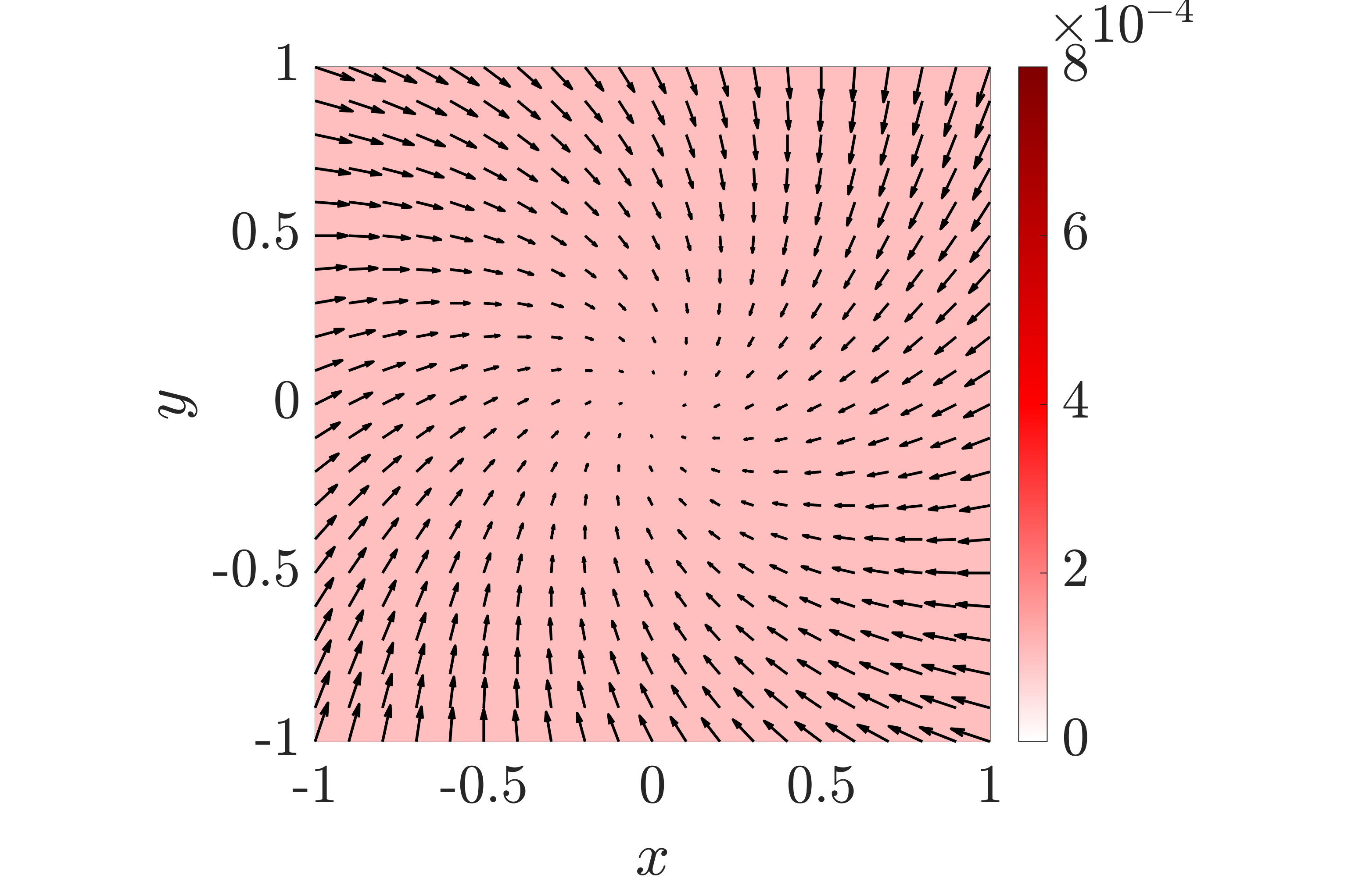}
    \includegraphics[width=0.49\textwidth]{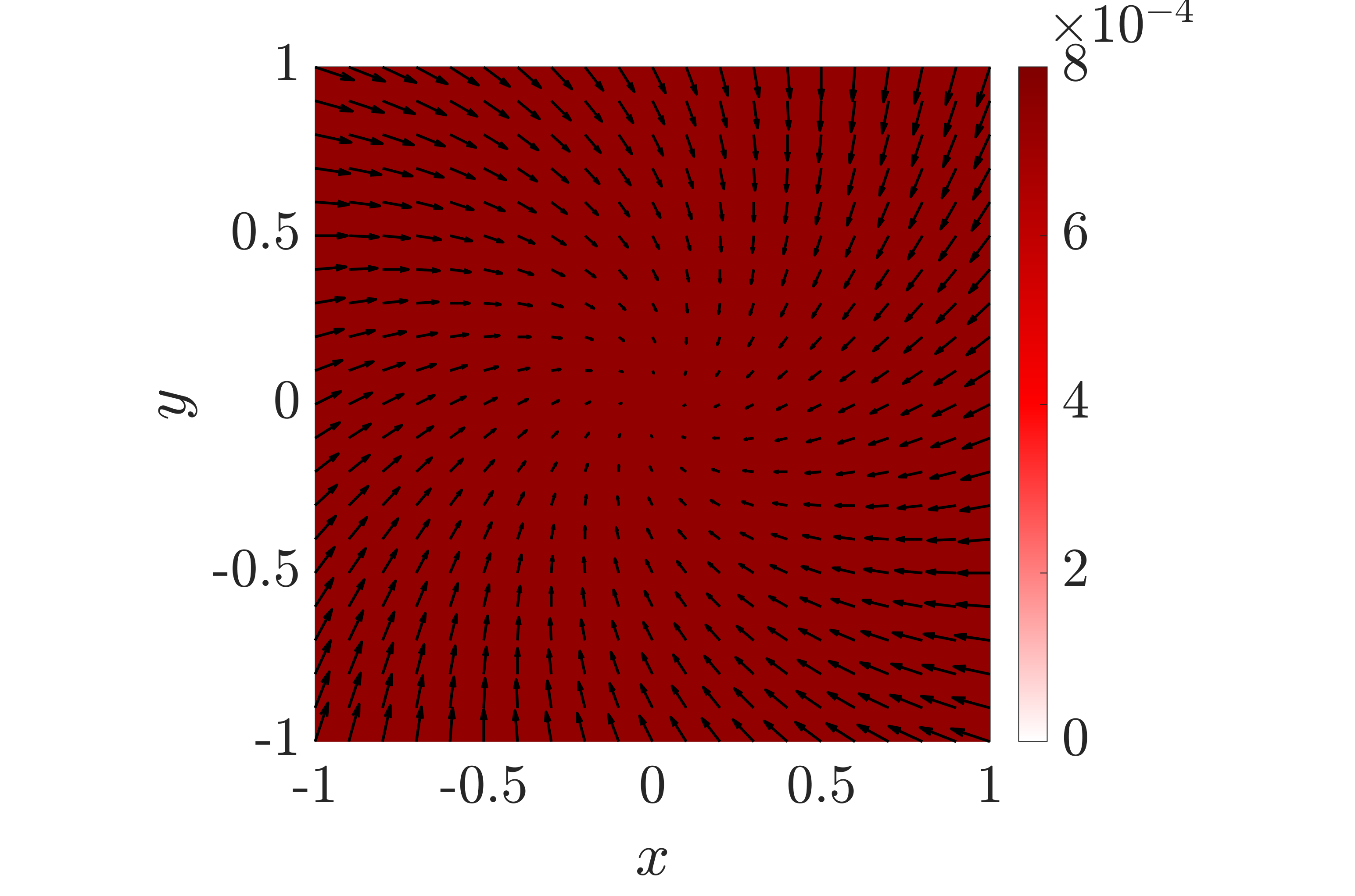}
    \caption{Velocity vector field $\boldsymbol{u}$ (arrow) and free surface height (contour) behaviour for Condition VII: (a) initial condition at $t=0$ and (b) exact solution at $t=1$ based on Theorem \ref{thm:theorem_exp_grow}. Parameters used include $F=1$, $\bar{f}=0.5$, $\bar{\tau}=1$, $D_0=0$, and $\tilde{h}_0=10^{-4}$ (Colour online).}
    \label{fig:exponential_growing}
\end{figure}

Figure \ref{fig:inertial_geostrophic_t_10000} illustrates an example of the behaviour of the exact solution corresponding to Condition I in which consequently follows Theorem \ref{thm:ADM_order_n_sin}. Figure \ref{fig:inertial_geostrophic_t_10000}(a) illustrates the initial behaviour of the geostrophic flow, where it is observed that the initial velocity vector $\boldsymbol{u}=\langle u(x,y,0),v(x,y,0)\rangle$ is orthogonal to the initial gradient of the free surface height which follows \eqref{eq:IC_cond_I}. Furthermore, Figure \ref{fig:inertial_geostrophic_t_10000}(b) illustrates that the long-term behaviour follows Corollary \ref{cor:linear_growing} where, for $t=10,000$, the velocity vector is consistent with \eqref{eq:inertial_geostrophic_inf_time_u} - \eqref{eq:inertial_geostrophic_inf_time_v}. Additionally, we notice that the free surface height has increased whereas the free surface gradient is unaffected which models equation \eqref{eq:inertial_geostrophic_inf_time_h}. Figure  \ref{fig:exponential_growing} illustrates the behaviour that corresponds to Condition VII, where we notice the increasing dynamics of the free surface height whereas the velocity vector field $\boldsymbol{u}$ is time-invariant which is consistent with  Theorem \ref{thm:theorem_exp_grow}.

\section{Discussion}
\label{sec:conclusion}

\par This work exploits the flexibility of the Adomian decomposition method (ADM) to connect the assumed forms that originated in the works of \cite{thacker1981some} and are still prevalent in the recent works of \cite{matskevich2019exact}. Therefore, we provided several additional generalisations of shallow-water wave phenomena depicting inertial geostrophic oscillations and anticyclonic vortices, while considering constant bottom friction, over flat bottom topographies. Consequently, our results demonstrate the importance of understanding the initial behaviour for these types of flows where we note the direct interplay between the Coriolis force $\bar{f}$, the constant bottom friction $\bar{\tau}$, and dissipation-induced instability from short and long-term standpoints. Hence, these results also  significantly advance the works of \cite{matskevich2019exact} and \cite{liu2021ADMSWE} where they can be applied to explore the effects of equatorial superrotation in atmospheric dynamics \citep{scott2007forced,scott2008equatorial}. Some avenues of future work include further generalising this approach to incorporate Beta plane approximations of the Coriolis force as well as exploring the effects of baroclinic shallow water waves while also exploring more generalized bottom topographies.  A comprehensive comparison between Adomian decomposition method against existing shallow water solvers \citep{couderc2013dassfow,delestre2017fullswof,yu2021hybrid} is another direction of future work.

\section*{Disclosure statement}

No potential conflict of interest was reported by the author(s).

\section*{Article Word Count}
3,404 words

\bibliography{main}
\bibliographystyle{gGAF}

\appendix
\section{Proof of Lemma \ref{ADM:PartialDerivativeCoeffLemma}}

\label{proof_ADM:PartialDerivativeCoeffLemma}

\begin{proof}
This is proven via mathematical induction by examining the recursion relationships for $u$, $v$, and $h$ in equation \eqref{eq:ADM_iter_Linv}. Condition \eqref{PartialDeriv_x_Cond} is demonstrated by examining the following relationships

\begin{equation}
\frac{\upartial^{2} u_{n+1}}{\upartial x^{2}}=-{L}_t^{-1}\left\{\frac{\upartial^{2}}{\upartial x^{2}}\left[A_n\left(u,\frac{\upartial u}{\upartial x}\right) + A_n\left(v,\frac{\upartial u}{\upartial y}\right) \right] + \frac{1}{F^2}\frac{\upartial^{3} h_{n}}{\upartial x^{3}} - \bar{f}\, \frac{\upartial^{2}v_n}{\upartial x^{2}}+\bar{\tau} \frac{\upartial^2 u_n}{\upartial x^2}\right\},
\label{un_xx}\nonumber
\end{equation}

\begin{equation}
\frac{\upartial^{2} v_{n+1}}{\upartial x^{2}} =-{L}_t^{-1}\left\{\frac{\upartial^{2}}{\upartial x^{2}} \left[A_n\left(u,\frac{\upartial v}{\upartial x}\right) + A_n\left(v,\frac{\upartial v}{\upartial y}\right) \right] + \frac{1}{F^2}\frac{\upartial^{3} h_{n}}{\upartial x^{2}\upartial y} + \bar{f} \, \frac{\upartial^{2} u_n}{\upartial x^{2}}+\bar{\tau}\frac{\upartial^2 v_n}{\upartial x^2} \right\},
\label{vn_xx}\nonumber
\end{equation}

\noindent and 

\begin{equation}
\frac{\upartial^{3} h_{n+1}}{\upartial x^{3}} = - {L}_t^{-1}\left\{\frac{\upartial^{4} }{\upartial x^{4}}[A_n(u,h)] + \frac{\upartial^{4} }{\upartial x^{3}\upartial y}[A_n(v,h)]+\frac{\upartial^{4} }{\upartial x^{4}}[u_n D] + \frac{\upartial^{4} }{\upartial x^{3}\upartial y}[v_n D]\right\}.
\label{hn_xxx}\nonumber
\end{equation}

\noindent When $n= 0$, the relationship between the initial and first components for $u$, $v$, and $h$ become

\begin{equation}
\frac{\upartial^{2} u_{1}}{\upartial x^{2}}=-{L}_t^{-1}\left\{\frac{\upartial^{2}}{\upartial x^{2}}\left[u_{0}\frac{\upartial u_{0}}{\upartial x} \,  + v_{0}\frac{\upartial u_{0}}{\upartial y} \right] + \frac{1}{F^2}\frac{\upartial^{3} h_{0}}{\upartial x^{3}} - \bar{f}\, \frac{\upartial^{2}v_{0}}{\upartial x^{2}}+ \bar{\tau} \frac{\upartial^2 u_0}{\upartial x^2}\right\},
\label{u1_xx}
\end{equation}

\begin{equation}
\frac{\upartial^{2} v_{1}}{\upartial x^{2}} =-{L}_t^{-1}\left\{\frac{\upartial^{2}}{\upartial x^{2}} \left[u_{0} \frac{\upartial v_{0}}{\upartial x}+ v_{0} \frac{\upartial v_{0}}{\upartial y} \right] + \frac{1}{F^2}\frac{\upartial^{3} h_{0}}{\upartial x^{2} \upartial y} + \bar{f} \, \frac{\upartial^{2} u_{0}}{\upartial x^{2}}+ \bar{\tau} \frac{\upartial^2 v_0}{\upartial x^2} \right\},
\label{v1_xx}
\end{equation}

\noindent and 

\begin{equation}
\frac{\upartial^{3} h_{1}}{\upartial x^{3}} = - {L}_t^{-1}\left\{\frac{\upartial^{4} }{\upartial x^{4}}\left[u_{0} h_{0} \right] + \frac{\upartial^{4} }{\upartial x^{3}y} \left[v_{0} h_{0} \right]+\frac{\upartial^{4} }{\upartial x^{4}} \left[u_{0} D \right] + \frac{\upartial^{4} }{\upartial x^{3}\upartial y} \left[v_{0} D \right]\right\}.
\label{h1_xxx}
\end{equation}

\noindent Employing \eqref{PartialDerivIC_x_Cond} - \eqref{PartialDerivIC_mixedxy_Cond} it can be shown that equations \eqref{u1_xx} - \eqref{h1_xxx} reduce to the following relationship

\begin{equation}
    \frac{\upartial^2 u_1(x,y,t)}{\upartial x^2} = \frac{\upartial^2 v_1(x,y,t)}{\upartial x^2} = \frac{\upartial^3 h_1(x,y,t)}{\upartial x^3} = 0. \nonumber
\label{PartialDerivComp1_x_Cond}
\end{equation}

\noindent Continuing this argument for $n \in \mathbb{N}^{+}$ yields equation  \eqref{PartialDeriv_x_Cond}. Similar arguments can be made to produce \eqref{PartialDeriv_y_Cond} and \eqref{PartialDeriv_mixedxy_Cond}, respectively.
\end{proof}

\section{Proof of Lemma \ref{ADM:PartialDerivativeCoeffLemma_sin}}
\label{proof_ADM:PartialDerivativeCoeffLemma_sin}
\begin{proof}
    This is proven via  mathematical induction by examining the recursion relationships for $u$, $v$, and $h$ in \eqref{eq:ADM_iter_Linv}. Condition \eqref{PartialDeriv_x_Cond_sin} is demonstrated by examining the following relationships

\begin{equation}
\frac{\upartial u_{n+1}}{\upartial x}=-{L}_t^{-1}\left\{\frac{\upartial}{\upartial x}\left[A_n\left(u,\frac{\upartial u}{\upartial x}\right) + A_n\left(v,\frac{\upartial u}{\upartial y}\right) \right] + \frac{1}{F^2}\frac{\upartial^{2} h_{n}}{\upartial x^{2}} - \bar{f}\, \frac{\upartial v_n}{\upartial x}+\bar{\tau} \frac{\upartial u_n}{\upartial x}\right\},
\label{un_xx_sin}\nonumber
\end{equation}

\begin{equation}
\frac{\upartial v_{n+1}}{\upartial x} =-{L}_t^{-1}\left\{\frac{\upartial}{\upartial x} \left[A_n\left(u,\frac{\upartial v}{\upartial x}\right) + A_n\left(v,\frac{\upartial v}{\upartial y}\right) \right] + \frac{1}{F^2}\frac{\upartial^{2} h_{n}}{\upartial x \upartial y} + \bar{f} \, \frac{\upartial u_n}{\upartial x}+\bar{\tau}\frac{\upartial v_n}{\upartial x} \right\},
\label{vn_xx_sarin}\nonumber
\end{equation}

\noindent and 

\begin{equation}
\frac{\upartial h_{n+1}}{\upartial x} = - {L}_t^{-1}\left\{\frac{\upartial^{2} }{\upartial x^{2}}[A_n(u,h)] + \frac{\upartial^{2} }{\upartial x\upartial y}[A_n(v,h)]\right\}.\nonumber
\label{hn_xxx_sin}
\end{equation}

\noindent For $n= 0$, the relationship between the initial and first components for $u$, $v$, and $h$ become

\begin{equation}
\frac{\upartial u_{1}}{\upartial x}=-{L}_t^{-1}\left\{\frac{\upartial}{\upartial x}\left[A_0\left(u,\frac{\upartial u}{\upartial x}\right) + A_0\left(v,\frac{\upartial u}{\upartial y}\right) \right] + \frac{1}{F^2}\frac{\upartial^{2} h_{0}}{\upartial x^{2}} - \bar{f}\, \frac{\upartial v_0}{\upartial x}+\bar{\tau} \frac{\upartial u_0}{\upartial x}\right\},
\label{u1_xx_sin}
\end{equation}

\begin{equation}
\frac{\upartial v_{1}}{\upartial x} =-{L}_t^{-1}\left\{\frac{\upartial}{\upartial x} \left[A_0\left(u,\frac{\upartial v}{\upartial x}\right) + A_0\left(v,\frac{\upartial v}{\upartial y}\right) \right] + \frac{1}{F^2}\frac{\upartial^{2} h_{0}}{\upartial x \upartial y} + \bar{f} \, \frac{\upartial u_0}{\upartial x}+\bar{\tau} \frac{\upartial v_0}{\upartial x} \right\},
\label{v1_xx_sin}
\end{equation}

\noindent and 

\begin{equation}
\frac{\upartial h_{1}}{\upartial x} = - {L}_t^{-1}\left\{\frac{\upartial^{2} }{\upartial x^{2}}[A_0(u,h)] + \frac{\upartial^{2} }{\upartial x\upartial y}[A_0(v,h)]\right\}.
\label{h1_xxx_sin}
\end{equation}

\noindent Employing \eqref{PartialDerivIC_x_Cond_sin} - \eqref{PartialDerivIC_mixedxy_Cond_sin} it can be shown that equations \eqref{u1_xx_sin} - \eqref{h1_xxx_sin} reduce to the following relationship

\begin{equation}
    \frac{\upartial u_1(x,y,t)}{\upartial x} = \frac{\upartial v_1(x,y,t)}{\upartial x} = \frac{\upartial h_1(x,y,t)}{\upartial x} = 0.
\label{PartialDerivComp1_x_Cond}
\end{equation}

\noindent Continuing this argument for $n \in \mathbb{N}^{+}$ yields equation  \eqref{PartialDeriv_x_Cond_sin}. Following similar arguments yields \eqref{PartialDeriv_y_Cond_sin}.
\end{proof}

\section{Proof of Corollary \ref{cor:linear_growing}}
\label{proof_cor:linear_growing}

\begin{proof}
The initial behaviour of $\tilde{u}_0$, $\tilde{v}_0$, and $\tilde{h}_0$ also satisfy equations \eqref{PartialDerivIC_x_Cond_sin} - \eqref{PartialDerivIC_mixedxy_Cond_sin} in which imply that $u, v$ and $h$ also satisfy \eqref{PartialCondition_u_sin} - \eqref{ParticalCondition_h_sinB}. Furthermore, we observe that equation \eqref{eq:ODE_sin} is also achieved where we note that the temporal dynamics of  $\tilde{u}_0(t)$ and $\tilde{v}_0(t)$, in equations \eqref{eq:ODEsinA} and \eqref{eq:ODEsinB}, do not directly depend on $\tilde{h}_0(t)$ whereas the temporal dynamics of $\tilde{h}_0(t)$, expressed in equation \eqref{eq:ODEsinC}, solely depend on $\tilde{u}_0(t)$ and $\tilde{v}_0(t)$. Hence, \eqref{eq:ODE_sin} can be analysed by first considering   

\begin{align}
     \frac{\rm d}{{\rm d} t}\boldsymbol{\psi}_{uv}(t)=&\boldsymbol{A}_{uv}\boldsymbol{\psi}_{uv}(t)+\boldsymbol{H}_{uv},
    \label{eq:eqaution_decaying_inertial_geostrophic_Corollary}
\end{align} 

\noindent where

\begin{align}
    \boldsymbol{\psi}_{uv}(t)=\begin{bmatrix}
    \tilde{u}_0(t)\\
    \tilde{v}_0(t)\end{bmatrix},\;\; \boldsymbol{\psi}_{uv}(0)=\begin{bmatrix}
    \tilde{u}_0(0)\\
    \tilde{v}_0(0)\end{bmatrix},\;\;\boldsymbol{A}_{uv}=\begin{bmatrix}
    -\bar{\tau} & \bar{f}\\
    -\bar{f} & -\bar{\tau}
    \end{bmatrix},\;\;\text{and}\;\;\boldsymbol{H}_{uv}=-\frac{1}{F^2}\begin{bmatrix}
    \tilde{h}_x\\
    \tilde{h}_y
    \end{bmatrix},\nonumber
\end{align}

\noindent which yields

\begin{align}
    \boldsymbol{\psi}_{uv}(t)=&{\rm e}^{\boldsymbol{A}_{uv}t}\boldsymbol{\psi}_{uv}(0)+\int_0^t {\rm e}^{\boldsymbol{A}_{uv}(t-\xi)}\boldsymbol{H}_{uv}{\rm d}\xi \nonumber\\
    =&{\rm e}^{\boldsymbol{A}_{uv}t}\boldsymbol{\psi}_{uv}(0)+\boldsymbol{A}_{uv}^{-1}(e^{\boldsymbol{A}_{uv}t}-\boldsymbol{I})\boldsymbol{H}_{uv},\nonumber
\end{align}

\noindent where  $\boldsymbol{A}_{uv}$ is invertible because $\bar{f}^2+\bar{\tau}^2\neq 0$ and

\begin{align}
    {\rm e}^{\boldsymbol{A}_{uv}t}:=\begin{bmatrix}
    {\rm e}^{-\bar{\tau} t}\cos(\bar{f}t)\;\; & {\rm e}^{-\bar{\tau} t}\sin(\bar{f}t)\\
    -{\rm e}^{-\bar{\tau} t}\sin(\bar{f}t)\;\; & {\rm e}^{-\bar{\tau} t}\cos(\bar{f}t)
    \end{bmatrix}.\nonumber
\end{align}

\noindent Therefore, to understand long-term behaviour it suffices to examine the behaviour as $t\rightarrow \infty$. Hence, 

\begin{align}
    \boldsymbol{\psi}_{uv}(\infty)=&\underset{t\rightarrow \infty}{\text{lim}}\boldsymbol{\psi}_{uv}(t)\nonumber\\
    =&-\boldsymbol{A}_{uv}^{-1}\boldsymbol{H}_{uv}\nonumber\\
    =&  \frac{1}{F^2(\bar{f}^2+\bar{\tau}^2)}\begin{bmatrix}
    -\bar{f}\tilde{h}_y-\bar{\tau} \tilde{h}_x\\
    \bar{f}\tilde{h}_x-\bar{\tau} \tilde{h}_y
    \end{bmatrix} \nonumber
    \label{eq:linear_growing_uv}
\end{align}

\noindent which yields \eqref{eq:inertial_geostrophic_inf_time_u} and  \eqref{eq:inertial_geostrophic_inf_time_v}. Next, we consider the dynamic equation for $\tilde{h}_0(t)$ given by 

\begin{equation}
    \frac{{\rm d}}{{\rm d}t}\tilde{h}_0(t)=-\tilde{h}_x\tilde{u}_0(t)  -\tilde{h}_y\tilde{v}_0(t),
\end{equation}

\noindent where $\tilde{h}_0(0)$ is the corresponding initial condition. Therefore, as $t \rightarrow \infty$ the temporal dynamics of $\tilde{h}_0(t)$ becomes

\begin{align}
    \underset{t\rightarrow \infty}{\text{lim}}\frac{{\rm d}}{{\rm d}t}\tilde{h}_0(t)=&-\tilde{h}_x\tilde{u}_0(\infty)-\tilde{h}_y\tilde{v}_0(\infty) \nonumber\\
    =&\frac{\bar{\tau} (\tilde{h}_x^2+\tilde{h}_y^2)}{F^2(\bar{f}^2+\bar{\tau}^2)} \nonumber
\end{align}

\noindent which results in 

\begin{equation}
    \underset{t\rightarrow \infty}{\text{lim}} \,\tilde{h}_0(t)=\frac{\bar{\tau} (\tilde{h}_x^2+\tilde{h}_y^2)t}{F^2(\bar{f}^2+\bar{\tau}^2)}+\tilde{h}_0(0),\nonumber
\end{equation}

\noindent and \eqref{eq:inertial_geostrophic_inf_time_h} is achieved.
\end{proof}

\section{Proof of Lemma \ref{ADM:PartialDerivativeCoeffLemma_tan_u}}
\label{proof_ADM:PartialDerivativeCoeffLemma_tan_u}

\begin{proof}
    This is proven via mathematical induction by examining the recursion relationships for $u$, $v$, and $h$ in equation \eqref{eq:ADM_iter_Linv}. Condition \eqref{PartialDeriv_x_Cond_tan_u_fy} is demonstrated by examining

\begin{equation}
 u_{n+1}=-{L}_t^{-1}\left\{\left[A_n\left(u,\frac{\upartial u}{\upartial x}\right) + A_n\left(v,\frac{\upartial u}{\upartial y}\right) \right] + \frac{1}{F^2}\frac{\upartial h_{n}}{\upartial x} - \bar{f}\, v_n +\bar{\tau} u_n\right\}.
\label{un_xx_tan_u}\nonumber
\end{equation}

\noindent For $n=0$ with \eqref{PartialDerivIC_x_Cond_tan_u_fy} - \eqref{PartialDerivIC_mixedxy_Cond_tan_u} achieves

\begin{align}
 u_{1}&=-{L}_t^{-1}\left\{\left[A_0\left(u,\frac{\upartial u}{\upartial x}\right) + A_0\left(v,\frac{\upartial u}{\upartial y}\right) \right] + \frac{1}{F^2}\frac{\upartial h_{0}}{\upartial x} - \bar{f}\, v_0+\bar{\tau} u_0\right\}\nonumber\\
 &=-{L}_t^{-1}\left\{u_0\frac{\upartial u_0}{\upartial x}+v_0 \frac{\upartial u_0}{\upartial y} - \bar{f}\,v_0 +\bar{\tau} u_0\right\}\nonumber\\
 &=0.\nonumber
\label{u1_xx_tan_u}
\end{align}

\noindent Continuing this argument for $n = \{1, 2, \ldots, n - 1\}$ yields equation  \eqref{PartialDeriv_x_Cond_tan_u_fy}.

Condition \eqref{PartialDeriv_x_Cond_tan_u} is demonstrated by examining the following relationships

\begin{equation}
\frac{\upartial^2 v_{n+1}}{\upartial x^2} =-{L}_t^{-1}\left\{\frac{\upartial^2}{\upartial x^2} \left[A_n\left(u,\frac{\upartial v}{\upartial x}\right) + A_n\left(v,\frac{\upartial v}{\upartial y}\right) \right] + \frac{1}{F^2}\frac{\upartial^{3} h_{n}}{\upartial x^2 \upartial y} + \bar{f} \, \frac{\upartial^2 u_n}{\upartial x^2 }+\bar{\tau} \frac{\upartial^2 v_n}{\upartial x^2} \right\},
\label{vn_xx_tan_u}\nonumber
\end{equation}

\noindent and 

\begin{equation}
\frac{\upartial h_{n+1}}{\upartial x} = - {L}_t^{-1}\left\{\frac{\upartial^{2} }{\upartial x^2}[A_n(u,h)] + \frac{\upartial^{2} }{\upartial x \upartial y}[A_n(v,h)]\right\}.
\label{hn_xxx_tan_u}\nonumber
\end{equation}

\noindent For $n= 0$, the relationship between the initial and first components for $v$, and $h$ become

\begin{equation}
\frac{\upartial^2 v_{1}}{\upartial x^2} =-{L}_t^{-1}\left\{\frac{\upartial^2}{\upartial x^2} \left[A_0\left(u,\frac{\upartial v}{\upartial x}\right) + A_0\left(v,\frac{\upartial v}{\upartial y}\right) \right] + \frac{1}{F^2}\frac{\upartial^{3} h_0}{\upartial x^2 \upartial y} + \bar{f} \, \frac{\upartial^2 u_0}{\upartial x^2}+\bar{\tau}\frac{\upartial^2 v_0}{\upartial x^2} \right\},
\label{v1_xx_tan_u}
\end{equation}

\noindent and 

\begin{equation}
\frac{\upartial h_{1}}{\upartial x} = - {L}_t^{-1}\left\{\frac{\upartial^{2} }{\upartial x^2}[A_0(u,h)] + \frac{\upartial^{2} }{\upartial x \upartial y}[A_0(v,h)]\right\}.
\label{h1_xxx_tan_u}
\end{equation}

\noindent Employing \eqref{PartialDerivIC_x_Cond_tan_u_fy} - \eqref{PartialDerivIC_mixedxy_Cond_tan_u}  equations \eqref{v1_xx_tan_u} - \eqref{h1_xxx_tan_u} become

\begin{equation}
   \frac{\upartial^2 v_1(x,y,t)}{\upartial x^2} = \frac{\upartial h_1(x,y,t)}{\upartial x} = 0.
\label{PartialDerivComp1_x_Cond}
\end{equation}

\noindent Continuing this argument for $n = \{1, 2, \ldots, n - 1\}$ yields equation  \eqref{PartialDeriv_x_Cond_tan_u}. Following similar arguments yields \eqref{PartialDeriv_y_Cond_tan_u} and \eqref{PartialDeriv_mixedxy_Cond_tan_u}.
\end{proof}

\section{Proof of Lemma \ref{ADM:PartialDerivativeCoeffLemma_tan_v}}
\label{proof_ADM:PartialDerivativeCoeffLemma_tan_v}

\begin{proof}
    This is proven via mathematical induction by examining the recursion relationships for $u$, $v$, and $h$ give by equation \eqref{eq:ADM_iter_Linv}. Condition \eqref{PartialDeriv_x_Cond_tan_v_fx} is demonstrated by examining the following relationships

\begin{equation}
v_{n+1} =-{L}_t^{-1}\left\{ \left[A_n\left(u,\frac{\upartial v}{\upartial x}\right) + A_n\left(v,\frac{\upartial v}{\upartial y}\right) \right] + \frac{1}{F^2}\frac{\upartial h_{n}}{\upartial y} + \bar{f} \,  u_n +\bar{\tau} v_n \right\}.
\nonumber
\end{equation}

\noindent For $n=0$ and using \eqref{PartialDerivIC_x_tan_v_fx} - \eqref{PartialDerivIC_mixedxy_Cond_tan_v}, yields
\begin{align}
v_{1}& =-{L}_t^{-1}\left\{ \left[A_0\left(u,\frac{\upartial v}{\upartial x}\right) + A_0\left(v,\frac{\upartial v}{\upartial y}\right) \right] + \frac{1}{F^2}\frac{\upartial h_{0}}{\upartial y} + \bar{f} \,  u_0+\bar{\tau} v_0 \right\}\nonumber\\
&=-{L}_t^{-1}\left\{u_0\frac{\upartial v_0}{\upartial x}+v_0\frac{\upartial v_0}{\upartial y} +\bar{f} \,  u_0+\bar{\tau} v_0 \right\}\nonumber\\
&=0.\nonumber
\label{u1_xx_tan_v_fx}
\end{align}

\noindent Employing similar arguments for $n={1,2,\ldots,n-1}$, achieves \eqref{PartialDeriv_x_Cond_tan_v_fx}. 

Condition \eqref{PartialDeriv_x_Cond_tan_v} is demonstrated by examining the following 

\begin{equation}
\frac{\upartial^{2} u_{n+1}}{\upartial x^{2}}=-{L}_t^{-1}\left\{\frac{\upartial^{2}}{\upartial x^{2}}\left[A_n\left(u,\frac{\upartial u}{\upartial x}\right) + A_n\left(v,\frac{\upartial u}{\upartial y}\right) \right] + \frac{1}{F^2}\frac{\upartial^{3} h_{n}}{\upartial x^{3}} - \bar{f}\, \frac{\upartial^{2}v_n}{\upartial x^{2}}+\bar{\tau} \frac{\upartial^2 u_n}{\upartial x^2}\right\},
\label{un_xx_tan_v}\nonumber
\end{equation}

\noindent and 

\begin{equation}
\frac{\upartial h_{n+1}}{\upartial x} = - {L}_t^{-1}\left\{\frac{\upartial^{2} }{\upartial x^{2}}[A_n(u,h)] + \frac{\upartial^{2} }{\upartial x \upartial y}[A_n(v,h)]\right\}.
\label{hn_xxx_tan_v}\nonumber
\end{equation}

\noindent For $n= 0$, the relationship between the initial and first components for $u$, and $h$ become

\begin{equation}
\frac{\upartial^{2} u_{1}}{\upartial x^{2}}=-{L}_t^{-1}\left\{\frac{\upartial^{2}}{\upartial x^{2}}\left[u_{0}\frac{\upartial u_{0}}{\upartial x} \,  + v_{0}\frac{\upartial u_{0}}{\upartial y} \right] + \frac{1}{F^2}\frac{\upartial^{3} h_{0}}{\upartial x^{3}} - \bar{f}\, \frac{\upartial^{2}v_{0}}{\upartial x^{2}}+\bar{\tau} \frac{\upartial^2 u_0}{\upartial x^2}\right\},
\label{u1_xx_tan_v}
\end{equation}

\noindent and 

\begin{equation}
\frac{\upartial h_{1}}{\upartial x} = - {L}_t^{-1}\left\{\frac{\upartial^{2} }{\upartial x^{2}}[A_0(u,h)] + \frac{\upartial^{2} }{\upartial x \upartial y}[A_0(v,h)]\right\}.
\label{h1_xxx_tan_v}
\end{equation}

\noindent Employing \eqref{PartialDerivIC_x_tan_v_fx} -   \eqref{PartialDerivIC_mixedxy_Cond_tan_v}, equations \eqref{u1_xx_tan_v} - \eqref{h1_xxx_tan_v} yield the following relationship

\begin{equation}
    \frac{\upartial^2 u_1(x,y,t)}{\upartial x^2} = \frac{\upartial h_1(x,y,t)}{\upartial x} = 0.
\label{PartialDerivComp1_x_Cond_tan_v}\nonumber
\end{equation}

\noindent Continuing this argument for $n = \{1, 2, \ldots, n - 1\}$ achieves   \eqref{PartialDeriv_x_Cond_tan_v}. Similar arguments yield \eqref{PartialDeriv_y_Cond_tan_v} and \eqref{PartialDeriv_mixedxy_Cond_tan_v}.
\end{proof}

\end{document}